\documentclass[10pt, conference, letterpaper]{IEEEtran}
\IEEEoverridecommandlockouts
\usepackage[cmex10]{amsmath}
\usepackage{cite}
\usepackage{amsmath, amsthm, amssymb}
\usepackage{algorithm}
\usepackage[noend]{algpseudocode}

\usepackage{float}
\usepackage{subfig}
\usepackage{graphicx}
\usepackage{mathtools}

\usepackage{array}
\newcolumntype{C}[1]{>{\centering\arraybackslash}m{#1}}

\newtheorem{theorem}{\textbf{Theorem}}

\newtheorem{lemma}{\textbf{Lemma}}

\newtheorem{eqs}{\textbf{Equation System}}

\newtheorem{corollary}{\textbf{Corollary}}

\theoremstyle{definition}

\newtheorem{problem}{\textbf{Problem}}
\newtheorem*{problem*}{\textbf{Problem}}
\newtheorem{definition}{\textbf{Definition}}
\newtheorem{remark}{\textbf{Remark}}

\let\G\relax

\DeclareMathOperator{\E}{\mathbb{E}}
\DeclareMathOperator{\g}{\text{\c{g}}}
\DeclareMathOperator{\G}{\mathcal{G}}

\DeclareMathOperator{\define}{\coloneqq}

\theoremstyle{definition}
\newtheorem{example}{\textbf{Example}}
\newtheorem{process}{\textbf{Process}}
\def\BibTeX{{\rm B\kern-.05em{\sc i\kern-.025em b}\kern-.08em
    T\kern-.1667em\lower.7ex\hbox{E}\kern-.125emX}}
\begin{document}

\title{An Approximation Algorithm for Active Friending in Online Social Networks}

\author{\IEEEauthorblockN{
Guangmo (Amo) Tong\IEEEauthorrefmark{1},
Ruiqi Wang\IEEEauthorrefmark{1},
Xiang Li\IEEEauthorrefmark{3},
Weili Wu\IEEEauthorrefmark{2}, 
and
Ding-Zhu Du\IEEEauthorrefmark{2}}\\
\IEEEauthorblockA{\IEEEauthorrefmark{1}Department of Computer and Information Sciences, University of Delaware, USA\\}
\IEEEauthorblockA{\IEEEauthorrefmark{2}Dept. of Computer Science, University of Texas at Dallas, USA} 
\IEEEauthorblockA{\IEEEauthorrefmark{3}Department of Computer Engineering, Santa Clara University, USA} 
\IEEEauthorblockA{\{amotong, wangrq\}@udel.edu, xli8@scu.edu,  \{weiliwu, dzdu\}@utdallas.edu}
}

\maketitle

\begin{abstract}
Guiding users to actively expanding their online social circles is one of the primary strategies for enhancing user participation and growing online social networks. In this paper, we study the active friending problem which aims at providing users with the strategy for methodically sending invitations to successfully build a friendship with target users. We consider the prominent linear threshold model for the friending process and formulate the active friending problem as an optimization problem. The key observation is the relationship between the active friending problem and the minimum subset cover problem, based on which we present the first randomized algorithm with a data-independent approximation ratio and a controllable success probability for general graphs. The performance of the proposed algorithm is theoretically analyzed and supported by encouraging simulation results done on extensive datasets.
\end{abstract}


\section{Introduction}
Due to the expeditious information exchange, the online social network has been heralded as the dominant platform for viral marketing \cite{kempe2003maximizing, nguyen2016stop}, news announcing \cite{lerman2010information}, and daily communication \cite{pempek2009college}. The success of social networks heavily relies on network growth in terms of the number of users, the intimacy of relationships, and the frequency of interactions. The recent decade has witnessed a tremendous expansion of online social network where there are totally 3.03 billion active users by the end of April 2018 \cite{smith2018amazing}. Facebook today has 2.1 billion users while this number was merely 12 million back to 2015. Strategies for network expansion can be classified into two categories: denotative expansion and connotative expansion \cite{yuan2017active}. Denotative expansion enlarges the network scale by attracting new users to create accounts, whereas connotative expansion aims at enhancing the network connectivity by fostering user interaction via methods such as friend recommendation. For example, the People-You-May-know widget is currently available on most online social networks. In this paper, we study the active friending problem which is one of the novel connotative expansion strategies.

Active friending is driven by the scenario when one user wishes for an online friendship with a target user who may be an influential person or a community leader  but not an acquaintance. Different from the traditional friending service which identifies potential contacts such as offline friends or friends of friends, active friending assists users to build point-to-point relationships rather than making selections among a pool of candidates. Even though users are free to send invitations to their target users, an invitation can hardly be accepted without the familiarity between users, especially when the target user is a celebrity who can receive many invitations. Alternatively, one promising method is to gain enough mutual friends with the target user before sending the invitation. The power of mutual friends has been observed for long, and in many social networks, such as Facebook and LinkedIn, common friends are displayed when an invitation is received. The active friending problem considered in this paper is to help user send invitations step-by-step to obtain a sufficient number of mutual friends with the target user and finally be an online friend of the target user.

\begin{figure*}[!t]
\centering
\includegraphics[width=0.8\textwidth]{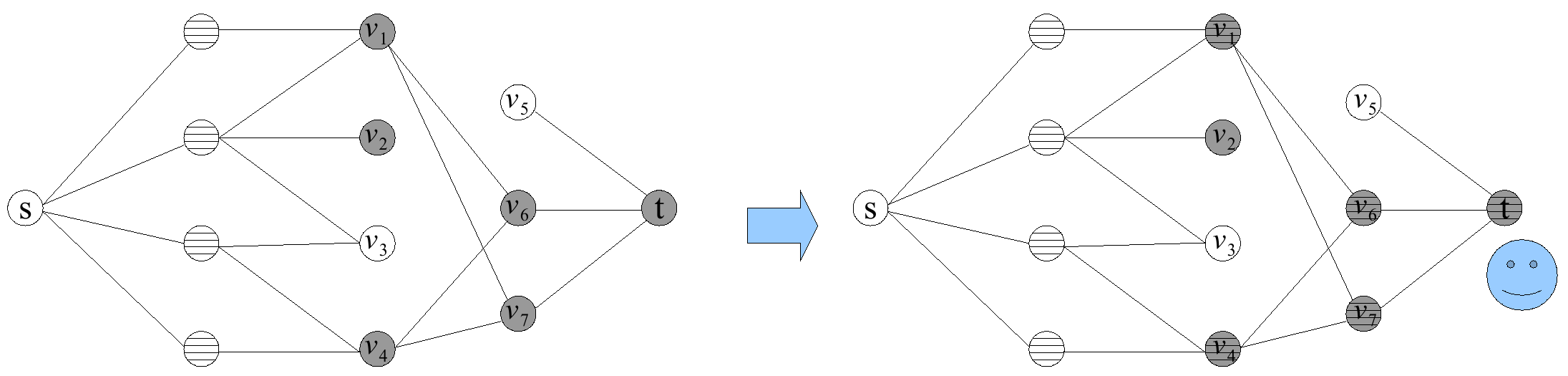}
\caption{An illustrative example of friending process.}
\label{fig: process}
\vspace{-3mm}
\end{figure*}

\textbf{Problem Formulation.} We denote the users $s$ and $t$ as the \textit{initiator} and \textit{target user}, respectively, when $s$ wishes to friend $t$. We adopt the threshold based friending model where the number of mutual friends is the major factor affecting the friending process. Therefore, the key issue is to have sufficient mutual friends before sending an invitation to $t$. To this end, the initiator $s$ has first to attempt to friend with the friends of $t$, which defines the approach recursively. Thus, the problem is to identify a set of intermediate users for $s$ to send invitations such that $t$ can finally be a friend of $s$. We denote the probability that $t$ can accept the invitation from $s$ as the \textit{acceptance probability}. In general, one can consider either the minimization version which is to find the minimum set of the intermediate users such that acceptance probability can reach a certain threshold, or the maximization version which is to maximize the acceptance probability under the size constraint of the invitations. In this paper, we will study the minimization version of the active friending problem. 

\textbf{The state-of-the-art.} The idea of active friending was first proposed by Yang \textit{et al.} \cite{yang2013maximizing} where the friending process was modeled based on the cascade model. In particular, they proposed the Selective Invitation with Tree and In Node Aggregation (SITINA) algorithm which derives the optimal solution to the maximum active friending problem when the underlying graph is approximated by a tree. Following this line, the authors in \cite{chen2014could} studied the same problem but considered the case when the network forms a DAG. Recently, the maximum active friending problem under the general graph was studied in \cite{yuan2017active} where an algorithm with a data-dependent\footnote{We say an approximation ratio is data-dependent if it depends on the social network structure and cannot be determined by only the number of nodes and edges.} approximation ratio was provided. In addition, it is shown in \cite{yuan2017active} that the active friending problem is NP-hard and the objective function is \#P-hard to compute. For the active friending problem, the existing works either provide heuristic algorithms or optimal solutions to special graphs, and to the best of our knowledge, there is no approximation algorithm available for neither the maximization version or minimization version on any of the popular operations models without assuming simplified graph structures. We in this paper make an attempt towards filling this gap by investigating the combinatorial structure behind the active friending problem under the threshold model.

\textbf{Contribution.} We study the minimum active friending problem under the threshold model and present a randomized algorithm with an approximation ratio of $O(\sqrt{n})$ where $n$ is the number of users. The proposed algorithm utilizes two ingredients: (a) a solution to the minimum subset cover problem to overcome the NP-hardness and (b) an estimating method to overcome the \#P-hardness. In addition to the theoretical analysis, the proposed algorithm consistently outperforms the trivial baseline methods, as shown in the simulations done on real-world social networks.

\textbf{Reproducibility.} The implementation of the proposed algorithm and the materials used in our experiments are made publicly available online \cite{git_activefriending}.

\textbf{Roadmap.} The preliminaries are provided in Sec. \ref{sec: pre}. The proposed algorithm and the theoretical analysis are shown in Sec. \ref{sec: algorithm}. In Sec. \ref{sec: exp}, we present the experimental settings and results. A brief survey of the related work is given in Sec. \ref{sec: related}. Sec. \ref{sec: conc} concludes this paper and discuss future work.

\section{Preliminaries}
\label{sec: pre}
\subsection{Model and Friending Process}
\label{subsec: model}
A snapshot of the social network is given by an undirected graph $G=(V,E)$ where $V$ and $E$ denote the user set and the current set of friendship, respectively. For two users $u$ and $v$, they are online friends iff $(u,v)\in E$. We use $n$ and $m$ to denote the number of users and edges, respectively. Associated with each ordered pair $(u,v)$ of users where $u$ and $v$ are friends, there is a weight $w_{(u,v)} \in (0,1]$ which characterizes the $v$'s familiarity with $u$. Note that $w_{(u,v)}$ is not necessarily equal to $w_{(v,u)}$ as the familiarity may not be symmetric. We use $N_v=\{u| (u,v)\in E\}$ to denote the current friends of user $v$. For the pair $u$ and $v$ where $u$ and $v$ are not friends, we explicitly set $w_{(u,v)}=w_{(v,u)}=0$. When two users $u$ and $v$ are not friends yet, $v$ is willing to accept the invitation from $u$ when they have enough mutual friends. In particular, each user $v$ is associated with a threshold $\theta_v$ and $v$ can accept the invitation from $u$ if $\sum_{v^{'} \in V^{'}} w_{(v^{'},v)} \geq \theta_v$ where $V^{'}$ is set of the mutual friends of $u$ and $v$. In order to handle the unobserved information, for each user $v$, we consider the $\theta_v$ uniformly selected from $[0,1]$ and assume $\sum_u w_{(u,v)} \leq 1$ after normalization. As aforementioned, we use $s$ and $t$ to denote the initiator and the target user, respectively. 

Suppose the set of the current friends of $s$ is $C$. We use $\Phi(C)$ to denote the set of the user who is not a friend of $s$ but is willing to be the friend of $s$. That is,
\begin{equation*}
\Phi(C)=\{u | u \notin C, \sum_{v \in N_u \cap C} w_{(v,u)} \geq \theta_u \}.
\end{equation*}
Since we have $w_{(v,u)}=0$ for the users $u$ and $v$ who are currently not friends, it is equivalent that 
\begin{equation}
\label{eq: phi}
\Phi(C)=\{u | u \notin C, \sum_{v \in C} w_{(v,u)} \geq \theta_u \}.
\end{equation}
For an invitation set $I \subseteq V$, the users in $I$ are called \textit{invited users}. Note that only the invited users can be the new friends of $s$. Given an invitation set $I \subseteq V$, the friending process goes round by round, shown as follows. 
\begin{process}
\label{processs: first}
Initially, $C_0(I)=N_s$ and the threshold of each user is randomly determined. Repeatedly obtain $C_{i+1}$ by 
\begin{equation}
\label{eq: update_C}
C_{i+1}(I)=C_{i}(I)\cup (\Phi(C_{i}(I)) \cap I),
\end{equation}
until $\Phi(C_i(I)) \cap I$ is empty or $t \in C_{i+1}(I)$. Let $C_{\infty}(I)$ be the $C_{i}(I)$ when the friending process terminates. $C_{\infty}(I)$ is in fact all the friends of $s$ under $I$, and therefore $t \in C_{\infty}(I)$ means the friending process is successful.
\end{process}

An example for illustration is shown below.

\begin{example}
\label{example: process}
Consider a network shown in Fig. \ref{fig: process} where $w_{(u,v)}=0.1$ for each ordered pair of users and suppose that the threshold of each user is $0.15$. Since $s$ and $t$ have no mutual friend at this time, inviting $t$ cannot  make $s$ successfully friend with $t$. Now suppose the invitation set is $\{v_{1},v_{2},v_{4},v_{6},v_{7}, t\}$. According to the process, $v_{1}$ and $v_{4}$ will be the first new friends of $s$, and, finally $v_{1}, v_{4}, v_{6}, v_{7}$ and $t$ will be the new friends of $s$. Note that $v_{3}$ could be the friend of $s$ but it does not receive an invitation, while $v_{2}$ receives an invitation but there are not enough mutual friends of $v_{2}$ and $s$.
\end{example}

\subsection{Minimum Active Friending}
For an invitation set $I \subseteq V$, we use $f(I)$ to denote the acceptance probability that $t$ can be a friend of $s$. In other words, $f(I)$ is the probability that $t$ appears in $C_{\infty}(I)$ under Process \ref{processs: first}.

\begin{remark}
The maximum value of $f(I)$ may not be one because the friending process does not necessarily succeed even if $I=V$. We use $p_{max}$ to denote the acceptance probability that is maximally possible. 
\end{remark}

We consider the following problem.

\begin{problem}[\textbf{Minimum Active Friending}]
\label{problem: mini}
Given a ratio $\alpha \in (0, 1]$, find an invitation set $I$ with the smallest size such that $f(I) \geq \alpha \cdot p_{max}$.
\end{problem}


\subsection{Minimum p-Union and Minimum Subset Cover}
Our algorithm for the active friending problem utilizes the existing results of the MpU problem.
\begin{problem}[\textbf{Minimum p-Union (MpU) Problem}]
\label{problem: MpU}
Given a set of elements $V$, a family $U$ of subsets of $V$ and an integer $p$, the MpU problem is to find a subset $U^{'}\subseteq U$ with $|U^{'}|=p$, such that $|\cup_{x \in U^{'}}x|$ is minimized.
\end{problem}
According to E. Chlamt\'{a}c \textit{et al.} \cite{chlamtac2018densest}, there exists a $(2\sqrt{|U|})$-approximation to the MpU problem. We denote this algorithm as the Chlamt\'{a}c algorithm, and we will take this algorithm as a subroutine to solve the active friending problem.

For a set of elements $V$ and two subsets $V_1, V_2 \subseteq V$, we say $V_1$ is \textit{covered} by $V_2$ iff $V_1 \subseteq V_2$. The minimum subset cover (MSC) problem is defined as follows.

\begin{problem}[\textbf{Minimum Subset Cover (MSC) Problem}]
\label{problem: msc}
Given a set of elements $V$, a family $U$ of subsets of $V$ and an integer $p$, find a subset $V^*$ of $V$ with the minimum cardinality such that at least $p$ subsets in $U$ are covered by $V^*$. 
\end{problem}

\begin{remark}
For any feasible solution $V^{'}$ covering $U^{'} \subseteq U$ with $|U^{'}|>p$, the union $V^{''}$ of any subset $U^{''}$ of $U^{'}$ with $|U^{''}|=p$ is also a feasible solution, and meanwhile $|V^{''}| =|\cup_{x \in U^{''}}|\leq |\cup_{x \in U^{'}}|\leq |V^{'}|$. To solve the MSC problem either optimally or approximately, it suffices to consider the subset of $V$ which is a union of exactly $p$ subsets of $U$, and consequently it is reduced to the MpU problem and the Chlamt\'{a}c algorithm provides a $(2\sqrt{|U|})$-approximation for the MSC problem. 
\end{remark}

\section{An Approximation Algorithm}
\label{sec: algorithm}
Now we are ready to present the algorithm for solving Problem \ref{problem: mini}. Our algorithm proceeds with two steps: (1) obtaining an unbiased estimator of the objective function by sampling; (2) maximizing the obtained estimator by using the Chlamt\'{a}c algorithm.
\subsection{An Unbiased Estimator of $f$}
We first introduce the preliminaries to construct the estimator of $f$. Note the friending process is in fact stochastic as the thresholds are generated randomly. The concept of \textit{realization} provides a derandomization of the friending process.

\begin{definition}[\textbf{Realization}]
\label{def: realization}
For a social network defined in Sec. \ref{subsec: model}, a realization is a mapping $g: V \rightarrow V$ randomly generated as follows. Each user $v$ randomly selects \textit{at most one} user among the initial friends where the friend $u \in N_v$ has the probability $w_{(u,v)}$ to be selected and with probability $1-\sum_{u \in N_v} w_{(u,v)}$ that $v$ selects no user. Define that 
\[
 g(v) \define
  \begin{cases}
  u &  \hspace{0mm} \hspace{-0.5mm} \text{if $v$ selects $u$} \\
  \aleph_0 & \hspace{0mm} \hspace{-0.5mm} \text{if $v$ selects no user} 
  \end{cases},
\]
where $\aleph_0 \notin V$ is an artificial user introduced for the purpose of analysis and $\aleph_0$ is not a friend of any user.
\end{definition}

We use $\mathcal{G}$ to denote the set of all possible realizations and let $\Pr[g]$ be the probability that $g \in \G$ can be generated. In addition, we use $\g$ to denote a random realization generated according to Def. \ref{def: realization}. The following process shows how to identify the new friends when the underlying realization is fixed.
\begin{process}
\label{process: realization}
For a realization $g$ and an invitation set $I \subseteq V$, we consider a set of nodes constructed step by step as follows. Initially, $H_{0}(g, I)= N_s$. Repeatedly obtain $H_{i+1}(g, I)$ by 
\begin{equation}
\label{eq: update_H}
H_{i+1}(g, I)=H_{i}(g, I) \cup (\Psi(H_{i}(g,I))\cap I)
\end{equation}
where 
\begin{equation}
\label{eq: psi}
\Psi(H_{i}(g,I))=\{v |~v\notin H_{i}(g, I),g(v)\in H_i(g, I)\} ,
\end{equation}
until $\Psi(H_{i}(g,I))\cap I=\emptyset$ or $t \in H_{i+1}(g,I)$. 
Let $H_{\infty}(g, I)$ be the set $H_{i}(g, I)$ when the process terminates. We use $f(g,I)$ to indicate that if $t$ belongs to $H_{\infty}(g, I)$, and $f(g,I)$ is define as
\[
f(g,I) \define
  \begin{cases}
  1 &  \hspace{0mm} \hspace{-0.5mm} \text{if $t \in H_{\infty}(g, I)$} \\
  0 & \hspace{0mm} \hspace{-0.5mm} \text{else } 
  \end{cases}
\]
Now let us take account of all the possible realizations and consider $\E[f(\g,I)]\define \sum_{g \in \G}\Pr[g] \cdot f(g,I)$. We use $H_{\infty}(\g,I)$ to denote the random set following the distribution: $\Pr[H_{\infty}(\g,I)=H_{\infty}(g,I)]=\Pr[g]$. Therefore, we have $\E[f(\g,I)]=\Pr[t \in H_{\infty}(\g,I)]$.
\end{process}

With an analysis similar to the one given in \cite{kempe2003maximizing}, we have the following result.
\begin{lemma}[Kempe \textit{et al.} \cite{kempe2003maximizing}]
\label{lemma: process}
$f(I)=\E[f(\g,I)]$.
\end{lemma}
\begin{proof}
The idea is to show that $C_\infty(I)$ and $H_{\infty}(\g,I)$ have the same distribution with respect to the realizations. Please see Appendix \ref{appendix: a} for a detailed proof.
\end{proof}




\begin{figure*}[t]
\centering
\subfloat[Case a.]{\label{fig: case_a}\includegraphics[width=0.40\textwidth]{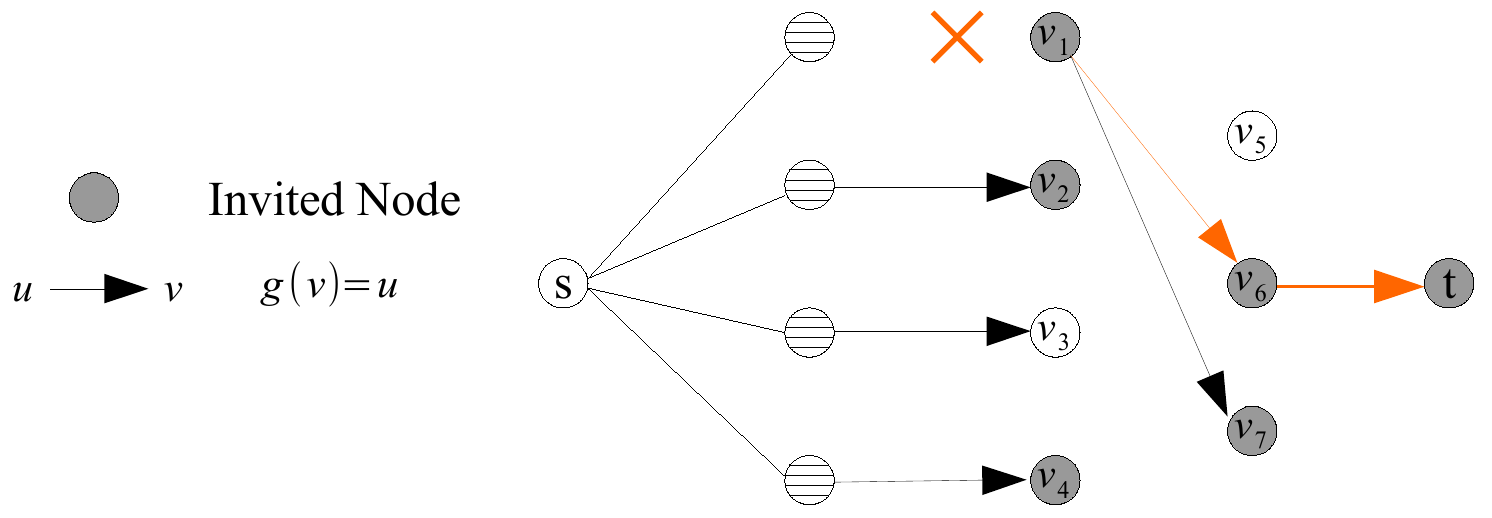}} 
\hspace{3mm}
\subfloat[Case b.]{\label{fig: case_b}\includegraphics[width=0.26\textwidth]{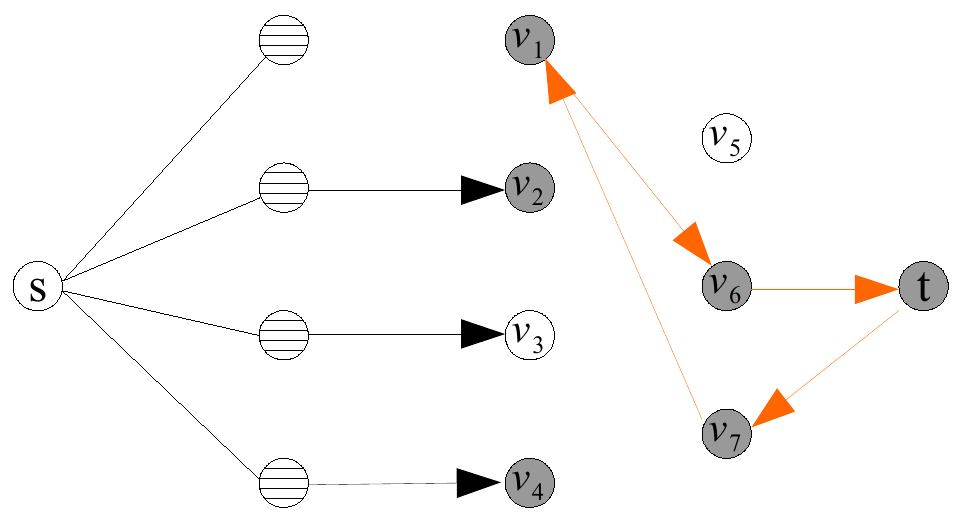}}  
\hspace{3mm}
\subfloat[Case c.]{\label{fig: case_c}\includegraphics[width=0.26\textwidth]{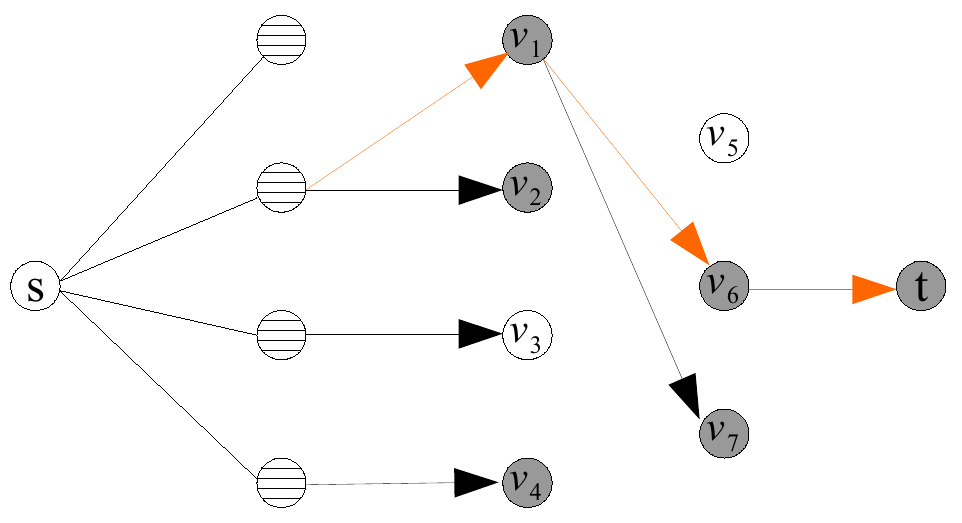}}  
\vspace{-1mm}
\caption{Cases in Lemma \ref{lemma: condition}.}
\label{fig: conditions}
\vspace{-3mm}
\end{figure*}

Next let us consider how to compute $f(g,I)$. That is, given an invitation set $I$, in which kind of realization $g$ that $t$ belongs to $H_{\infty}(g, I)$? It turns out that we do not have to generate the whole set $H_{\infty}(g, I)$ by Process \ref{process: realization}. Instead, it suffices to consider a user set $t(g)$ identified by Alg. \ref{alg: t(g)}. As shown in Alg. \ref{alg: t(g)}, we track the user back according to $g$ starting from the target $t$, and add the encountered users to $t(g)$, until no new node can be further reached or a node in $N_s$ has been reached. For each realization $g$ and invitation set $I$, we say $I$ \textit{covers} $g$ iff $t(g) \subseteq I$. The following is a key lemma showing the condition for $t$ to be a friend of $s$.

\begin{algorithm}[t]
\caption{t(g)}\label{alg: t(g)}
\begin{algorithmic}[1]
\State \textbf{Input:} $g, t$ and $N_s$;
\State \textbf{Output:} a user set $t(g)$;
\State $t(g) \leftarrow\{t\}$, $u^* \leftarrow t$;
\While {true};
\small
\If {$g(u^*) =\aleph_0 $} $t(g) \leftarrow t(g) \cup \{\aleph_0\} $ and return $t(g)$;
\EndIf
\If  {$g(u^*) \in t(g)$ } $t(g) \leftarrow t(g) \cup \{\aleph_0\} $ and return $t(g)$;
\EndIf
\If {$g(u^*) \in N_s$} return $t(g)$;
\EndIf
\State $t(g) \leftarrow t(g) \cup \{g(u^*)\}$ and $u^* \leftarrow g(u^*)$;
\EndWhile
\end{algorithmic}
\end{algorithm}

\begin{lemma}
\label{lemma: condition}
For each realization $g$ and invitation set $I$, $t$ can be a friend of $s$ in $g$ if and only if $I$ covers $g$. 
\end{lemma}
\begin{proof}
According to Def. \ref{def: realization}, it is useful to imagine a realization as a directed graph where $(u,v)$ exists iff $g(v)=u$. The users connected to $t$ forms a path because each user can select at most one user among their initial friends. This path is exactly $t(g)$. According to Process \ref{process: realization}, $t$ can be a friend of $s$ if and only if there is a path from $N_s$ to $t$ where all the nodes in the path are invited. There are three cases to consider, as illustrated in Fig. \ref{fig: conditions}.
\begin{itemize}
\item Case a. Before reaching any user in $N_s$, the path ends with some user selecting no user (line 5 in Alg. \ref{alg: t(g)}).  
\item Case b. It forms a cycle (line 6 in Alg. \ref{alg: t(g)}). 
\item Case c. The path reaches some user $u^*$ in $N_s$ (line 7 in Alg. \ref{alg: t(g)}). 
\end{itemize}
Because the users are finite, each realization must be in one of the above three cases. For the first two cases, according to Process \ref{process: realization}, any invitation set $I$ cannot make $t$ be a friend of $s$, and accordingly, $t(g)$ contains the artificial user $\aleph_0$ so it cannot be a subset of any invitation set $I \subseteq V$. For the third case, since only an invited user can be a friend of $s$, $I$ should cover all the users through the path from $u^*$ to $t$ except $u^*$, i.e., $t(g) \subseteq I$. Thus, proved.
\end{proof}

The following result immediately follows from Lemmas \ref{lemma: process} and \ref{lemma: condition}.

\begin{corollary}
\label{coro: mean_I}
For each invitation set $I \subseteq V$, $\E[f($\emph{\c{g}}$,I)]=f(I)=\sum_{g \in \G}\Pr[g]f(g,I)$, where
\begin{equation}
\label{eq: f(I)}
f(g,I) =
  \begin{cases}
  1 &  \hspace{0mm} \hspace{-0.5mm} \text{if $t(g) \subseteq I$} \\
  0 & \hspace{0mm} \hspace{-0.5mm} \text{else } 
  \end{cases}.
\end{equation}
\end{corollary}




\ 
\begin{definition}[Type-1/0 Realization]
For convenience, we say a realization $g$ is a type-1 realization if $\aleph_0 \notin t(g)$. Otherwise, we say $g$ is a type-0 realization. Furthermore, we use the binary value $y(g)$ to denote the type of realization $g$, and $y(g)$ is defined as \[y(g) \define
\begin{cases}
1 &  \hspace{0mm} \hspace{-0.5mm} \text{if $g$ is type-1} \\
0 & \hspace{0mm} \hspace{-0.5mm} \text{else } 
\end{cases}\]
\end{definition}

We can see that $t$ cannot be a friend of $s$ under a type-0 realization even if we send invitation to all the users. Therefore, $f(g,V)=1$ if and only if $y(g)=1$,. Following Corollary \ref{coro: mean_I}, we have the following result showing that $y(\g)$ is an unbiased estimator of $p_{max}$.

\begin{algorithm}[t]
\caption{Estimating $p_{max}$}\label{alg: estimate}
\begin{algorithmic}[1]
\State \textbf{Input:} $\epsilon$ and $N$;
\State $\Upsilon \leftarrow 1+\frac{4(e-2)(1+\epsilon)\ln(2/N)}{\epsilon^2}$;
\State $i \leftarrow 0$, $j\leftarrow 0$
\While {$j \leq \Upsilon$}
\State $i \leftarrow i+1$;
\State Generate a realization $g$ by Alg. \ref{alg: algorithm}.
\State $j \leftarrow j+y(g)$;
\EndWhile
\Return $\Upsilon/i$;
\end{algorithmic}
\end{algorithm}

\begin{corollary}
\label{coro: mean_p_max}
$\E[y($\emph{\c{g}}$)]=f(V)=p_{max}.$
\end{corollary}

Since $y(\g)$ is an unbiased estimator of $p_{max}$, the standard Monte Carlo Estimation can be applied to estimating $p_{max}$. The algorithm is shown in Alg. \ref{alg: estimate}. We have the following result due to \cite{dagum2000optimal}. 

\begin{lemma}[Dagum \textit{et al.} \cite{dagum2000optimal}]
\label{lemma: p_max^*}
For each $0 \leq \epsilon \leq 1$ and $N \geq 0$, there exists an algorithm which produces a $p_{max}^*$ such that \[
\Pr\Big[|p_{max}^*-p_{max}| \leq \epsilon \cdot p_{max}\Big] \geq 1-1/N,\] where the number of the used simulations is asymptotically bounded by 
\begin{equation}
l_0 \define \frac{\epsilon_0^2+4(e-2)(1+\epsilon_0)\ln(N/2)}{\epsilon^2_0 \cdot p_{max}}.
\end{equation}
\end{lemma}

\textbf{An Idea.} According to Corollary \ref{coro: mean_I}, $\sum_{g \in \G}\Pr[g] \cdot f(g,I)$ is in fact an explicit formula of $f(I)$. However, it is not feasible to directly maximize it because its value cannot be efficiently computed as there are exponential number of realizations in $\G$. Alternatively, we consider a set $B_l=\{\g_1,...,\g_l\}$ of $l$ random realizations each of which is generated independently at random. We partition the realizations in $B_l$ into two subsets $B_l^0$ and $B_l^1$ where $B^0_l=\{\g \in B|~y(\g)=0\}$ and $B_l^1=\{\g \in B|~y(\g)=1\}$ are the sets of the type-0 realizations and type-1 realizations, respectively. For each set $B_l$ of realizations and $I \subseteq V$, define that 
\begin{equation*}
F(B_l, I) \define \sum_{g \in B_l}f(g,I).
\end{equation*}
Note that $F(B_l, V)=|B_l^1|$, and we will use $F(B_l, V)$ and $|B_l^1|$ interchangeably. According to Corollaries \ref{coro: mean_I} and \ref{coro: mean_p_max}, $|F(B_l, V)/l-p_{max}|$ and $|F(B_l, I)/l-f(I)|$ can be arbitrarily small provided that $l$ is sufficiently large. As a result, for an invitation set $I$ satisfying $F(B_l, I)\geq \alpha \cdot F(B_l, V)$, $f(I)\geq \alpha \cdot p_{max}$ should be ensured with a high probability when $l$ is sufficiently large. Furthermore, it is desired to find the $I$ with the minimum cardinality such that $F(B_l, I)\geq \alpha \cdot F(B_l, V)$. Finally, because type-0 realization cannot be covered by any invitation set, it suffices to consider the type-1 realizations in $B_l$. Thus, this is equivalent to solving the following problem.
\begin{problem}
Given a collection $B_l^1$ of type-1 realizations, and an integer $p \leq |B_l^1|$, find a subset $V^*$ of $V$ with the minimum cardinality such that at least $p$ realizations in $B_l^1$ are covered, i.e., $F(B_l, V^*) \geq p$.
\end{problem}

We can easily check that this problem can be reduced to the MSC problem with the input $V, \{t(g_1),...,t(g_{|B_l^1|})\}$ and $p$. Therefore, the Chlamt\'{a}c algorithm can produce an invitation set $I^*$ such that for each $I^{'} \subseteq V$ with $F(B_l, I^{'})\geq p$,
\begin{equation}
\label{eq: alpha}
F(B_l, I^*) \geq p,
\end{equation}
and
\begin{equation}
\label{eq: chlamtac}
|I^*| \leq 2\sqrt{|B_l^1|} \cdot |I^{'}|.
\end{equation}

{
\begin{algorithm}[t]
\caption{The framework}\label{alg: algorithm}
\begin{algorithmic}[1]
\State \textbf{Input:} $\beta$ and $l$;
\State Generate $l$ random realizations $B_l$ and let $B_l^1 \subseteq B_l$ be the set of the realization(s) $g$ with $y(g)=1$.
\State Solve the MSC problem approximately by the Chlamt\'{a}c algorithm with input $V, \{t(g_1),...,t(g_{|B_l^1|})\}$ and $\lceil \beta \cdot |B_l^1| \rceil$. Let $I^*$ be the output.
\State Return $I^*$;
\end{algorithmic}
\end{algorithm}}

\subsection{The Algorithm}
\textbf{Framework.} The framework is shown in Alg. \ref{alg: algorithm} with two parameters $\beta$ and $l$. We first generate $l$ realizations among which we denote the set of the type-1 realizations as $B_l^1=\{g_1,...,g_{|B_l^1|}\}$. And then obtain an invitation set $I^*$ by running the Chlamt\'{a}c algorithm with input: $V, \{t(g_1),...,t(g_{|B_l^1|})\}$ and $\beta \cdot |B_l^1|$. The rest of this section aims at determining $l$ and $\beta$ such that the performance can be guaranteed.

We use $I_{\alpha}$ to denote the optimal solution to Problem \ref{problem: mini} associated with the input $\alpha$, and let $I^*$ be the solution produced by Alg. \ref{alg: algorithm}. In addition, let $0 <\epsilon< \alpha$ and $N>0$ be two parameters which are used to control the performance. Since the algorithm is randomized, our goal is to find an invitation set $I^*$ such that, with probability at least $1-1/N$, we have $f(I^*) \geq (\alpha-\epsilon) \cdot p_{max}$ and meanwhile $|I^*|/|I_{\alpha}|$ can be bounded by a provable factor. Throughout this section, we assume $\epsilon$ and $N$ are fixed.





We use the following centrality inequalities to analyze the accuracy of the estimations. Let $X_i \in [0,1]$ be $l$ i.i.d random variables where $E(X_i)=\mu$. For each $\delta > 0$, the Chernoff bound \cite{motwani2010randomized} states that 
\begin{equation}
\label{eq:chernoff_1}
\mathrm{Pr}\Big[|\sum X_i - l \cdot \mu| \geq \delta \cdot l \cdot  \mu  \Big]\leq 2 \exp(-\frac{l \cdot \mu \cdot \delta^2}{2+\delta})
\end{equation}



\textbf{A Sufficient Condition.} Let $\epsilon_0$, $\epsilon_1$ $\in (0,1)$ be some parameters that will be determined later, and $p_{max}^*$ be an estimate of $p_{max}$ obtained by Monte Carlo simulation. Suppose a set $B_l=\{\g_1,...,\g_l\}$ of $l$ random realizations is used in Alg. \ref{alg: algorithm}, and let $B_l^1$ be the set of the type-1 realizations in $B_l$. The following equation system will be sufficient to ensure the desired performance guarantees.

\begin{eqs}
\label{eqs:1}
\begin{align}
\label{eq: eqs_accu_pmax}
&|p_{max}^*-p_{max}| \leq \epsilon_0 \cdot p_{max}  \\
\label{eq: eqs_accu_I}
&|F(B_l, I)/l-f(I)|\leq \epsilon_1 \cdot p_{max}^*, \text{~for each~} I \subseteq V  \\
\label{eq: eqs_beta}
&\beta=\frac{\alpha-\epsilon_1\cdot (1+\epsilon_0)}{1+\epsilon_1\cdot (1+\epsilon_0)}>0 \\
\label{eq: eqs_epsilon}
&\beta \cdot \big(1-\epsilon_1\cdot (1+\epsilon_0)\big)-\epsilon_1\cdot (1+\epsilon_0)=\alpha-\epsilon 
\end{align}
\end{eqs}

First, Equation System 1 ensures the objective value is bounded.
\begin{lemma}
\label{lemma: prob}
With Equation System \ref{eqs:1}, $f(I^*) \geq (\alpha-\epsilon)\cdot p_{max}$.
\end{lemma}
\begin{proof}
By Eqs. (\ref{eq: eqs_accu_I}) and (\ref{eq: alpha}), we have 
\begin{equation*}
f(I^*) \geq F(B_l, I^*)/l-\epsilon_1 \cdot p_{max}^*,
\end{equation*}
and
\begin{equation*}
F(B_l, I^*) \geq \beta \cdot |B_l^1|=\beta \cdot F(B_l, V).
\end{equation*}
Putting the above together, we have
\begin{equation}
\label{eq: lemma_prob_1}
f(I^*) \geq \beta \cdot F(B_l, V)/l-\epsilon_1 \cdot p_{max}^*.
\end{equation}
On the other hand, applying Eq. (\ref{eq: eqs_accu_I}) to $I=V$, we have 
\begin{equation*}
F(B_l, V)/l \geq f(V)-\epsilon_1 \cdot p_{max}^*= p_{max}-\epsilon_1 \cdot p_{max}^*,
\end{equation*}
and combining Eq. (\ref{eq: lemma_prob_1}), we have 
\begin{equation*}
f(I^*) \geq \beta \cdot (p_{max}-\epsilon_1 \cdot p_{max}^*)-\epsilon_1 \cdot p_{max}^*
\end{equation*}
Furthermore, due to Eq. (\ref{eq: eqs_accu_pmax}), we have \[f(I^*) \geq \beta \cdot (p_{max}-\epsilon_1 \cdot (1+\epsilon_0)\cdot p_{max})-\epsilon_1 \cdot (1+\epsilon_0)\cdot p_{max}.\] Finally, because of Eq. (\ref{eq: eqs_epsilon}), we have $f(I^*) \geq (\alpha-\epsilon)\cdot p_{max}$. 
\end{proof}

Second, Equation System 1 ensures the size of the solution is bounded.
\begin{lemma}
\label{lemma: size}
With Equation System \ref{eqs:1}, $|I^*| \leq 2\sqrt{|B_l^1|}\cdot |I_{\alpha}|$.
\end{lemma}
\begin{proof}
Note that $I^*$ is obtained by the Chlamt\'{a}c algorithm with the input $V, \{t(g_1),...,t(g_{|B_l^1|})\}$ and $\lceil \beta \cdot |B_l^1| \rceil$. By Eq. (\ref{eq: chlamtac}), it suffices to show that $F(B_l, I_{\alpha}) \geq \beta \cdot |B_l^1|$. Applying Eq. (\ref{eq: eqs_accu_I}) to $I=I_{\alpha}$, we have \[F(B_l, I_{\alpha})/l \geq f(I_{\alpha})-\epsilon_1 \cdot p_{max}^*.\] Because $I_{\alpha}$ is the optimal solution to Problem \ref{problem: mini}, we have $I_{\alpha} \geq \alpha \cdot p_{max}$ and therefore, 
\[F(B_l, I_{\alpha})/l \geq \alpha \cdot p_{max}-\epsilon_1 \cdot p_{max}^*\] Combining Eq. (\ref{eq: eqs_accu_pmax}), we further have 
\begin{equation}
\label{eq: temp_1}
F(B_l, I_{\alpha})/l \geq (\alpha-\epsilon_1 \cdot (1+\epsilon_0)) \cdot p_{max}.
\end{equation}
On the other hand, applying Eq. (\ref{eq: eqs_accu_I}) to $I=V$, we have \[F(B_l, V)/l-f(V) \leq \epsilon_1 \cdot p_{max}^* \leq \epsilon_1 \cdot (1+\epsilon_0)\cdot p_{max},\]
where the last inequality follows from Eq. (\ref{eq: eqs_accu_pmax}). Since $f(V)=p_{max}$, we have 
\begin{equation*}
p_{max}\geq \frac{F(B_l, V)}{l\cdot (1+\epsilon_1 \cdot (1+\epsilon_0))}=\frac{|B_l^1|}{l\cdot (1+\epsilon_1 \cdot (1+\epsilon_0))}.
\end{equation*}
Combining Eqs. (\ref{eq: temp_1}) and (\ref{eq: eqs_beta}), this implies that 
\begin{equation*}
F(B_l, I_{\alpha}) \geq \beta \cdot |B_l^1|.
\end{equation*}
Thus, proved.
\end{proof}

According to the above two lemmas, we have the desired performance guarantee provided that Equation System \ref{eqs:1} is satisfied. 

\textbf{Making Equation System \ref{eqs:1} Satisfied.} Due to Lemma \ref{lemma: p_max^*}, an estimate $p_{max}^*$ satisfying Eq. (\ref{eq: eqs_accu_pmax}) is obtainable. Furthermore, there exist $\epsilon_0$ and $\epsilon_1$ that are able to make Eqs. (\ref{eq: eqs_beta}) and (\ref{eq: eqs_epsilon}) satisfied, because the LHS of Eq. (\ref{eq: eqs_epsilon}) approaches to $\alpha$ when $\epsilon_0$ and $\epsilon_1$ approach to 0. In addition, $\beta$ is given by $\epsilon_0$ and $\epsilon_1$. Thus, the only part left to consider is Eq. (\ref{eq: eqs_accu_I}). According to Corollary \ref{coro: mean_I}, Eq. (\ref{eq: eqs_accu_I}) can be satisfied if $l$ is sufficiently large. In particular, a threshold is given in the next lemma.

\begin{lemma}
\label{lemma: l}
With probability at least $1-1/N$, $|F(B_l, I)/l-f(I)|\leq \epsilon_1 \cdot p_{max}^*$ holds for each $I \subseteq V$, if $|p_{max}^*-p_{max}| \leq \epsilon_0 \cdot p_{max}$ and $l \geq l^*$ where 
\begin{equation}
\label{eq: l^*}
l^* \define \frac{(\ln 2+\ln N+n\ln 2)\cdot(2+\epsilon_1\cdot (1-\epsilon_0))}{\epsilon_1^2 \cdot (1-\epsilon_0)^2 \cdot p_{max}^*}.
\end{equation}
\end{lemma}
\begin{proof}
For a certain subset $I \subseteq V$, by the Chernoff bound, 
\[\mathrm{Pr}\Big[|F(B_l, I)/l-f(I)| \geq \frac{\epsilon_1 \cdot p_{max}^*}{f(I)} \cdot f(I)  \Big]\] 
is no larger than \[2\exp(-\frac{l \cdot  \epsilon_1^2 \cdot (p_{max}^*)^2}{2f(I)+\epsilon_1 \cdot  p_{max}^*}).\] Because $p_{max}^* \geq(1-\epsilon_0) \cdot p_{max}$ and $l \geq l^*$, this probability is no larger than $\frac{1}{N2^n}$. Note that there are $2^n$ subsets of $V$, Due to the union bound, with probability at least $1-1/N$, \[|F(B_l, I)/l-f(I)|\leq \epsilon_1 \cdot p_{max}^*\] holds simultaneously for all the subsets.
\end{proof}

\textbf{RAF Algorithm.} Given $\epsilon$ and $N$, the whole process consists of three steps: (1) determine $\epsilon_0$ and $\epsilon_1$ such that Eqs. (\ref{eq: eqs_beta}) and (\ref{eq: eqs_epsilon}) are satisfied; (2) obtain an estimate $p_{max}^*$ of $p_{max}$ by Lemma \ref{lemma: p_max^*} with $\epsilon_0$ and $N$; (3) obtain an invitation set by Alg. \ref{alg: algorithm} with $\beta$ and $l=l^*$. We denote this algorithm as the Realization-based Active Friending (RAF) algorithm.

\begin{algorithm}[t]
\caption{RAF algorithm}\label{alg: whole}
\begin{algorithmic}[1]
\State \textbf{Input:} $\alpha$, $\epsilon$ and $N$;
\State Determine $\epsilon_0$, $\epsilon_1$ and $\beta$ by solving Eq. (\ref{eq: new_eqsys}).
\State Obtain $p_{max}^*$ by Alg. \ref{alg: estimate} with $\epsilon_0$ and $N$.
\State Compute $l^*$ according to Eq. (\ref{eq: l^*});
\State $I^{*} \leftarrow$ Alg. \ref{alg: algorithm} with $\beta$ and $l^*$;
\State Return $I^*$;
\end{algorithmic}
\end{algorithm}

\begin{theorem}
\label{theorem: main}
With probability at least $1-2/N$, the RAF algorithm outputs an invitation set $I^*$ such that $f(I^*) \geq (\alpha-\epsilon)\cdot p_{max}$ and $|I^*|/|I_{\alpha}|= O(\sqrt{n})$. 
\end{theorem}
\begin{proof}
By Lemma \ref{lemma: p_max^*}, $|p_{max}^*-p_{max}| \leq \epsilon_0 \cdot p_{max}$ is valid with probability at least $1-1/N$. Under this condition, By Lemma \ref{lemma: l}, Eq. (\ref{eq: eqs_accu_I}) is satisfied with probability at least $1-1/N$. Therefore, the Eqs. \ref{eqs:1} holds with probability at least $1-2/N$ and we have the results in Lemmas \ref{lemma: prob} and  \ref{lemma: size}. Note that $|B_l^1|=F(B_l,V)$. By Eqs. (\ref{eq: eqs_accu_pmax}) and (\ref{eq: eqs_accu_I}), \[F(B_l,V)\leq l \cdot (\frac{1}{1+\epsilon_0}+\epsilon_1)\cdot p_{max}^*.\] Because (a) $l=l^*$ and (b) $\epsilon_0$ and $\epsilon_0$ are independent of $n$ and $N$, $F(B_l,V)=O(n)$ for each fixed $\epsilon$ and $N=O(n^K)$ for each $K \in \mathbb{Z}^+$. Thus, $|I^*|/|I_{\alpha}|= O(\sqrt{n})$.
\end{proof}

The only part left is the selection of $\epsilon_0$ and $\epsilon_1$. According to Eq. (\ref{eq: eqs_epsilon}), $\epsilon_1$ becomes relatively small when $\epsilon_0$ is relatively large, which means the time used to estimate $p_{max}$ becomes shorter and the time taken by Alg \ref{alg: algorithm} becomes longer due to the increase in $l^*$. Therefore, there is a trade-off between the running time of  step 2 and step 3. We adopt the setting that $\epsilon_0=n \cdot \epsilon_1$ such that the running time of step 2 and step 3 have the same asymptotic order with respect to $n$, and therefore, we obtain $\epsilon_0$ and $\epsilon_1$ by solving
\begin{equation}
\label{eq: new_eqsys}
\begin{aligned}
&\epsilon_0=n\cdot \epsilon_1\\
&\beta\cdot \big(1-\alpha(1+\epsilon_1)\big)-\epsilon_1\cdot (1+\epsilon_0)=\alpha-\epsilon\\
&\beta=\frac{\alpha-\epsilon_1\cdot (1+\epsilon_0)}{1+\epsilon_1\cdot (1+\epsilon_0)} \\
\end{aligned}
\end{equation}

\begin{table*}[t]
\renewcommand{\arraystretch}{1.8}
\centering
{\begin{tabular}{| C{2.7cm} || C{2.7cm}| C{2.7cm}| C{2.7cm}| C{2.7cm} |}
\hline
\textbf{} 	& \textbf{Wiki} 		& \textbf{HepTh }& \textbf{HepPh} &  \textbf{Youtube} \\   
\hline
\hline
\textbf{nodes \#} 	& 7K		& 28K		& 35K &  1.1M\\
\hline
\textbf{edges \#}	& 103K		& 353K	& 421K &  6.0M \\
\hline
\textbf{Avg. Degree}	& 14.7		& 12.6	& 12.0  &  5.54 \\
\hline
\end{tabular}}
\vspace{3mm}
\caption{Datasets}
\label{table: datasets}
\vspace{3mm}
\end{table*}

The whole process is formally given in Alg. \ref{alg: whole}. Because the Chlamt\'{a}c algorithm is polynomial and the time for generating a realization is $O(m)$, the whole algorithm is polynomial.

\begin{remark}
The $t(\g)$ of a random realization $\g$ can be generated by the reverse sampling approach proposed by C. Borgs \textit{et al.} \cite{borgs2014maximizing}, enabling us to avoid sampling every edge in the graph and thus enhance practical efficiency. However, the worst-case is still $O(m)$.
\end{remark}


\subsection{A special case: $\alpha=1$}
\label{subsec: special}
Though Problem \ref{problem: mini} is hard to solve for the general case, it is polynomial-time solvable when we are looking for the the invitation set achieving $p_{max}$ (i.e., $\alpha=1$). Clearly we have $p_{max}=f(V)$ but we are interested the minimum set $I \subseteq V$ such that $f(I)=p_{max}$. Let $V_{max} \subseteq V$ be the set of nodes where a node $u$ is in $V_{max}$ iff $u$ is on some path from a node in $\{s\}\cup N_s$ to $t$ and $u \notin \{s\}\cup N_s$. It turns out that $V_{max}$ is minimum set resulting in the maximum acceptance probability, as shown in the next lemma.

\begin{lemma}
\label{lemma: mini_p_max}
$f(V_{max})=p_{max}$ and $V_{max}$ is the unique minimum invitation set that achieves $p_{max}$.
\end{lemma}
\begin{proof}

First, a node $u$ not in $\{s\}\cup N_s\cup V_{max}$ cannot be in the set $t(g)$ for any realization $g$, and therefore $V_{max} \cup \{u\}$ cannot cover more realizations than $V_{max}$ does. Thus, according to Eq. (\ref{eq: f(I)}), $p_{max}=f(V_{max})$.  Second, for any node $u$ in $V_{max}$, it belongs to at least one path from some node in $\{s\}\cup N_s$ to $t$, and therefore it must be in the $t(g)$ for some realization $g$ in $\G$ with $\aleph_0 \notin t(g)$. Consequently, due to the removal of $u$, $V_{max} \setminus \{u\}$ fails to cover at least one realization that was covered by $V_{max}$. As a result, again according to Eq. (\ref{eq: f(I)}), $f(V_{max}\setminus \{u\})<f(V_{max})$, which implies that $V_{max}$ is the minimum set that achieves $p_{max}$.
\end{proof}
Since $V_{max}$ can be computed by the simple graph search, the problem is polynomial-time solvable. Furthermore, since each $I^*$ produced by Alg. \ref{alg: algorithm} must be a subset of $V_{max}$, Theorem \ref{theorem: main} is still valid if we replace the $n$ in Eq. (\ref{eq: l^*}) by $|V_{max}|$, which slightly reduces the running time of Alg. \ref{alg: whole}.



\section{Performance Evaluation}
\label{sec: exp}
In this section, we present the experiments for evaluating the proposed algorithm. Our experiments were performed on a server with a 3.6 GHz quad-core processor.

\textbf{Datasets.} We considered four social network datasets borrowed from J. Leskovec \cite{snapnets}, of which the statistics are listed in Table \ref{table: datasets}. \textit{Wiki} is a small who-votes-on-whom social network collected from Wikipedia. \textit{HepTh} and \textit{HepPh} are two middle-size citation networks of Arxiv.org from the categories of High Energy Physics Phenomenology and High Energy Physics Theory, respectively. \textit{Youtube} is a large social network collected Youtube.com.

\textbf{Friending Model.} Following the convention \cite{kempe2003maximizing}, we consider the setting where $w_{(u,v)}=1/|N_v|$. This setting has been widely considered in the prior work (e.g., \cite{chen2010scalable,yuan2017active,kempe2003maximizing}). 


\textbf{Problem Setting.} For each dataset, we randomly select $500$ pairs of $s$ and $t$ with $p_{max}$ no less than $0.01$ and report the average results. The threshold of $p_{max}$ helps in ruling out the case when $p_{max}$ is extremely small. Such cases are not interesting as the friending process is pessimistic even if we sent invitations to all the other users, implying that no strategy can be effective. The value $p_{max}$ is estimated by Monte Carlo simulation for each pair of $s$ and $t$. 

\textbf{Baseline Algorithms.} Notice that the existing algorithms provided by Yuan \textit{et al.} \cite{yuan2017active}, Yang \textit{et al.} \cite{yang2013maximizing} and Chen \textit{et al.} \cite{chen2009make} are designed for the maximum active friending problem, whereas RAF solves its minimization version. 
Therefore, we compare the RAF algorithm with two popular heuristics, Shortest Path (SP) algorithm and High Degree (HD) algorithm. When selecting invited nodes, HD prefers the node with the highest degree while SP prefers the nodes on the shortest path from $s$ to $t$. The solutions given by these three algorithms are denoted as, $I_{RAF}$, $I_{HD}$ and $I_{SP}$, respectively.

We conducted a series of experiments. In the first experiment, we compare the performance of RAF, HD and SP when the size of the invitation set is fixed where the size of the invitation set is given by the solution from RAF. In the second experiment, we compare RAF, HD and SP by examining the number of invited nodes they need to reach a certain friending probability. In the third experiment, we compare the solution given by RAF with $V_{max}$ to examine the input-output ratio. These experiments are presented in the following subsections.

\begin{figure*}[t]
\centering
\subfloat[Wiki]{\label{fig: wiki_exp_1}\includegraphics[width=0.24\textwidth]{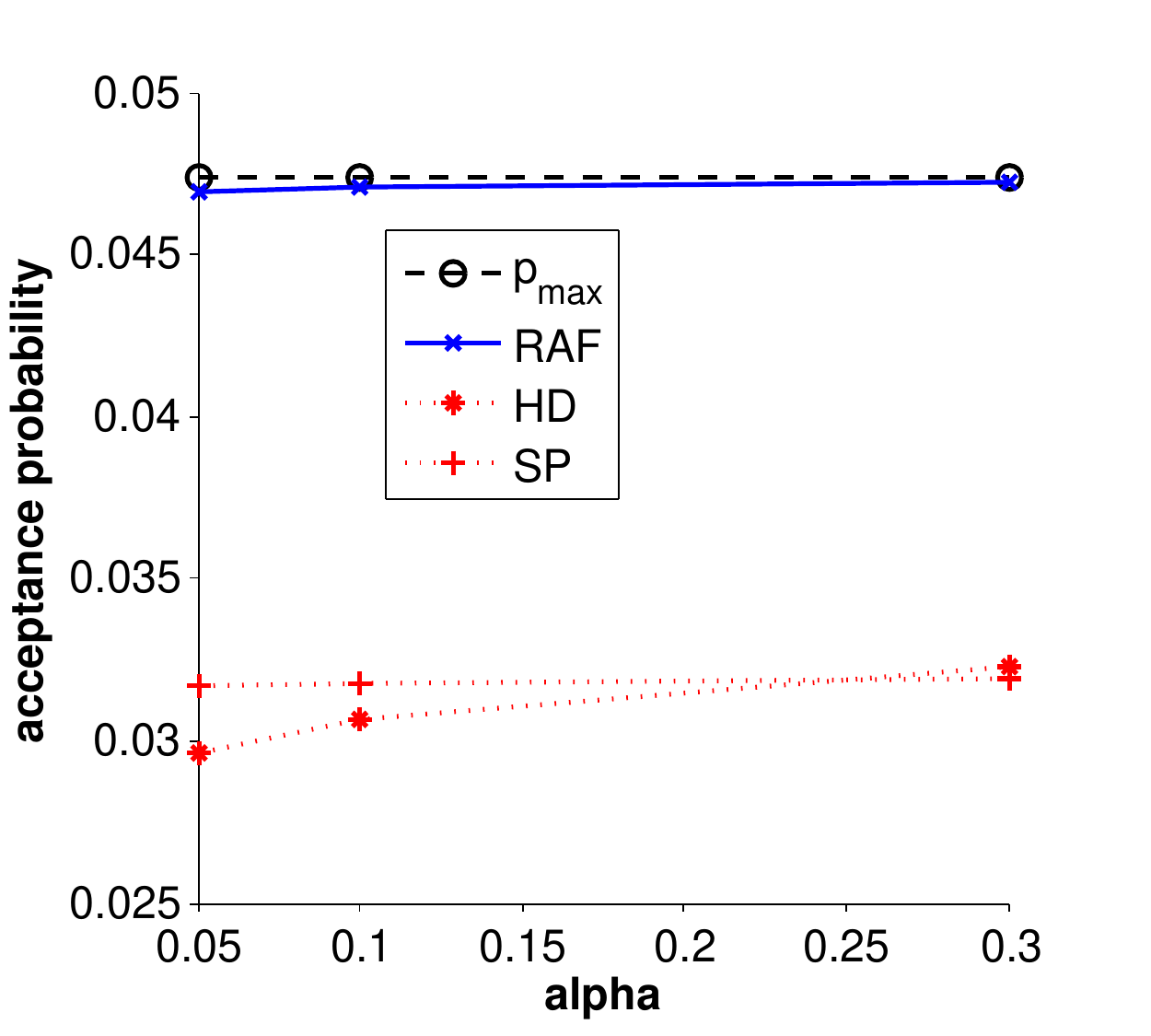}} 
\subfloat[HepPh]{\label{fig: hepph_exp_1}\includegraphics[width=0.24\textwidth]{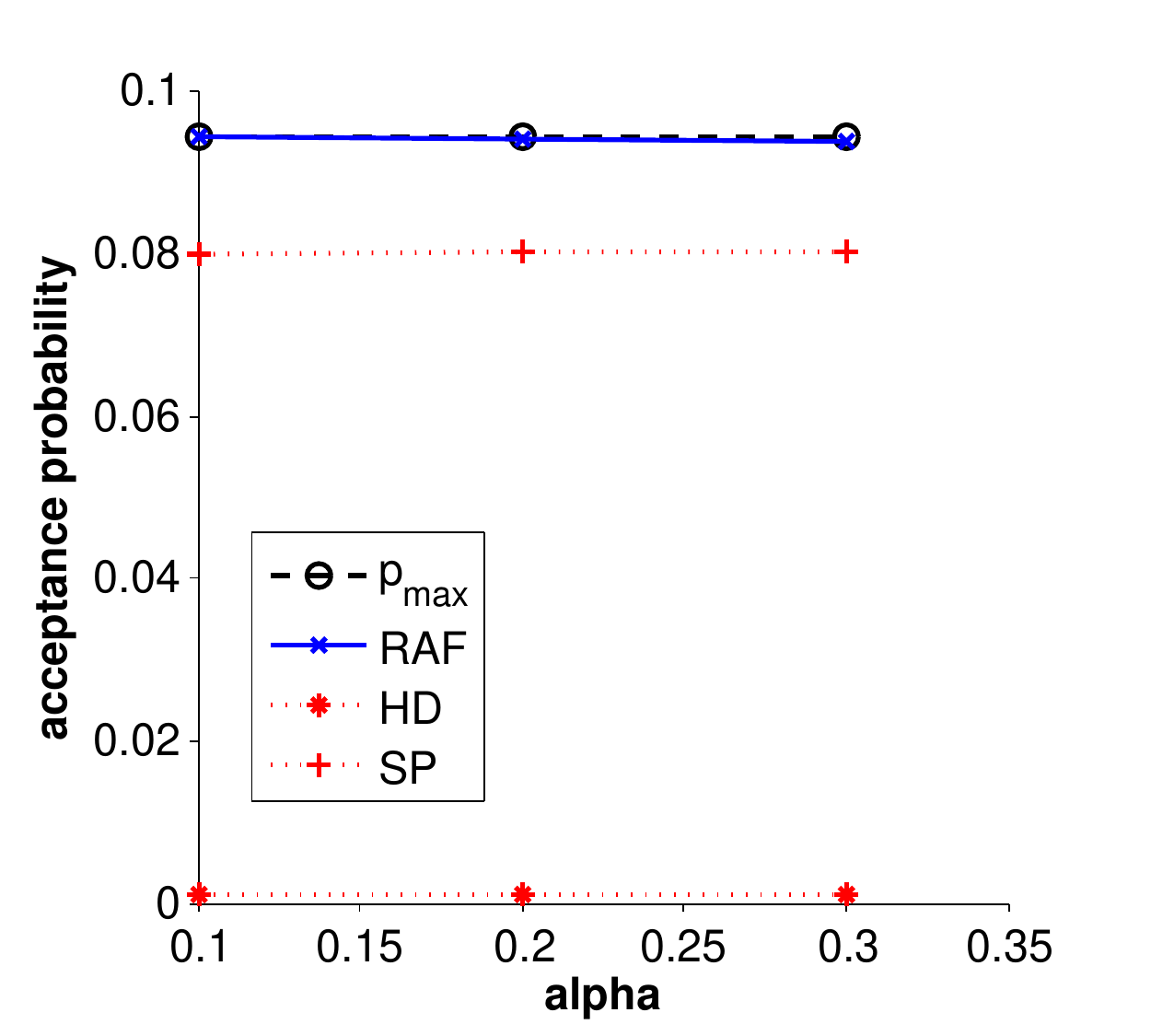}} 
\subfloat[HepTh]{\label{fig: hepth_exp_1}\includegraphics[width=0.24\textwidth]{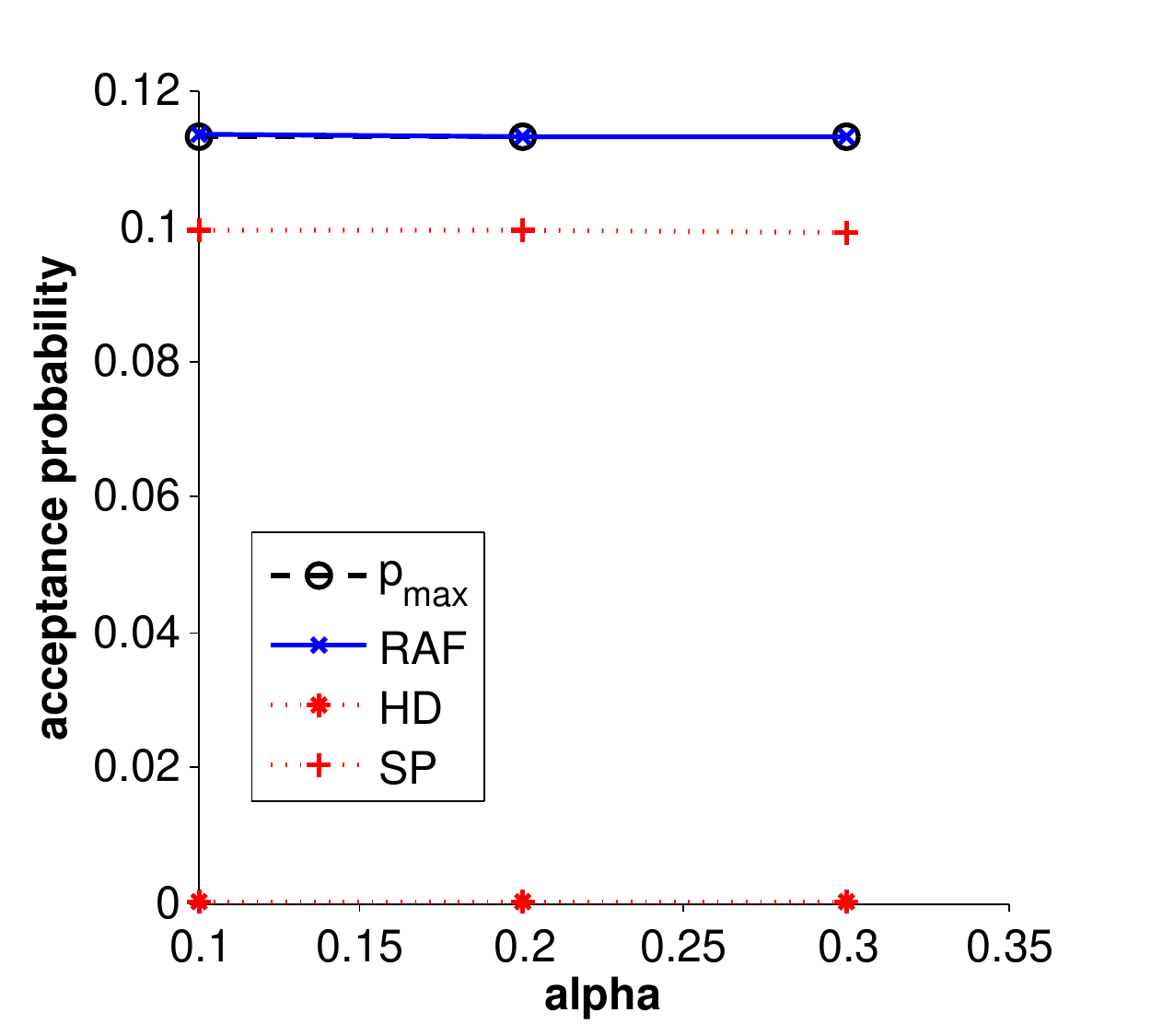}}
\subfloat[Youtube]{\label{fig: youtube_exp_1}\includegraphics[width=0.24\textwidth]{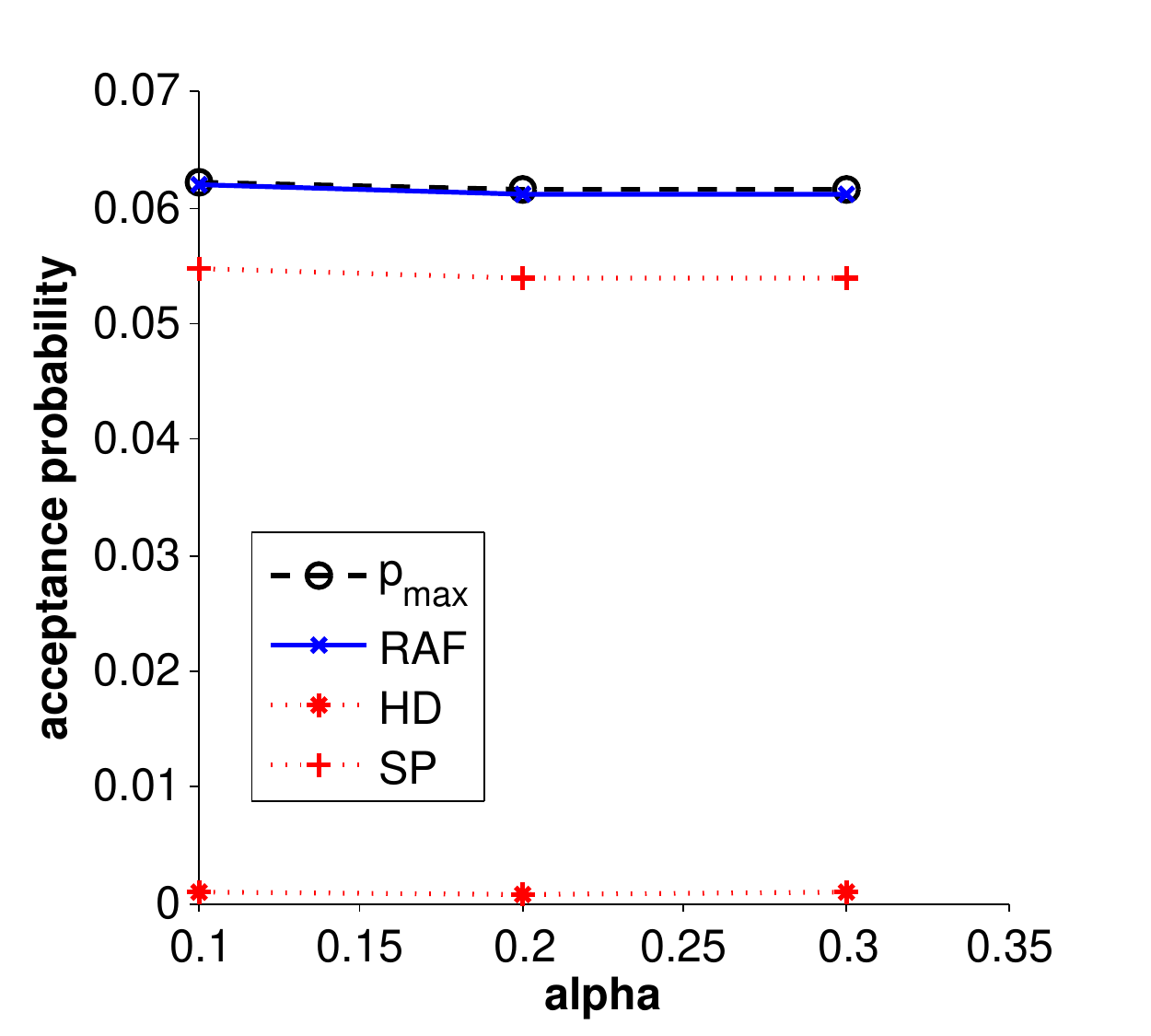}}
\caption{Basic experiment}
\vspace{3mm}
\label{fig: exp}
\end{figure*}

\begin{figure*}[t]
\centering
\subfloat[Wiki]{\label{fig: wiki_hd}\includegraphics[width=0.24\textwidth]{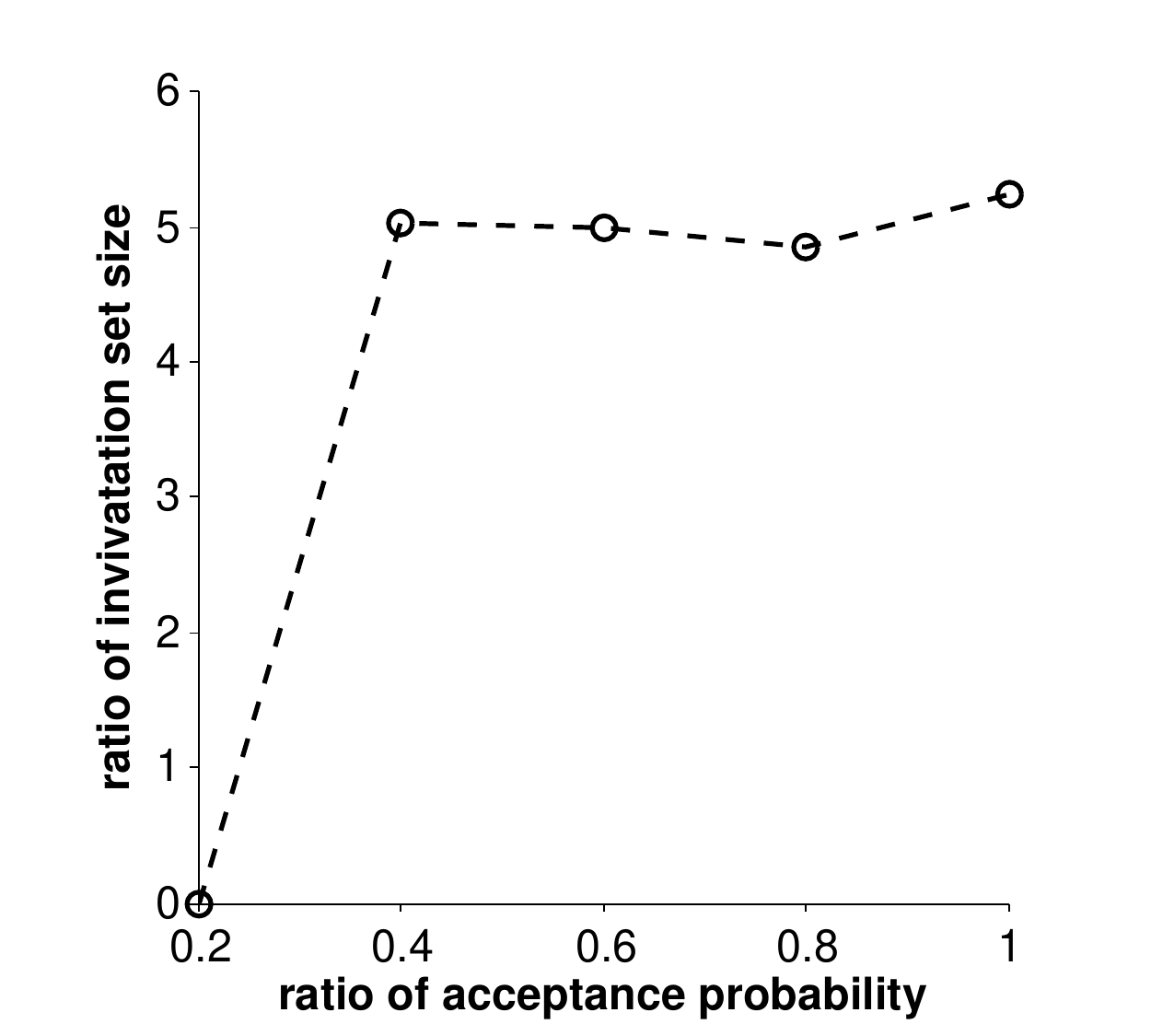}} 
\subfloat[HepPh]{\label{fig: hepph_hd}\includegraphics[width=0.24\textwidth]{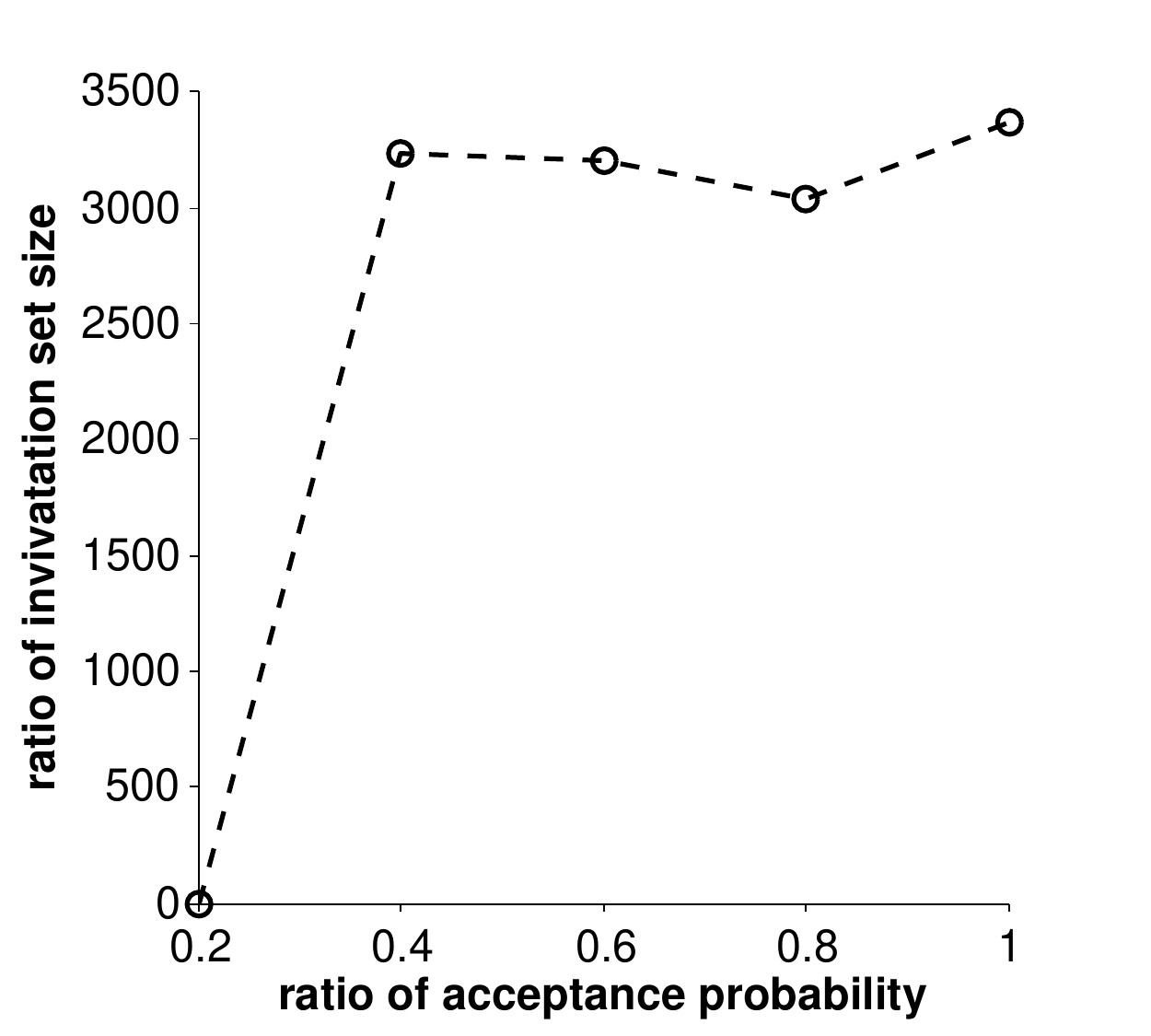}} 
\subfloat[HepTh]{\label{fig: hepth_hd}\includegraphics[width=0.24\textwidth]{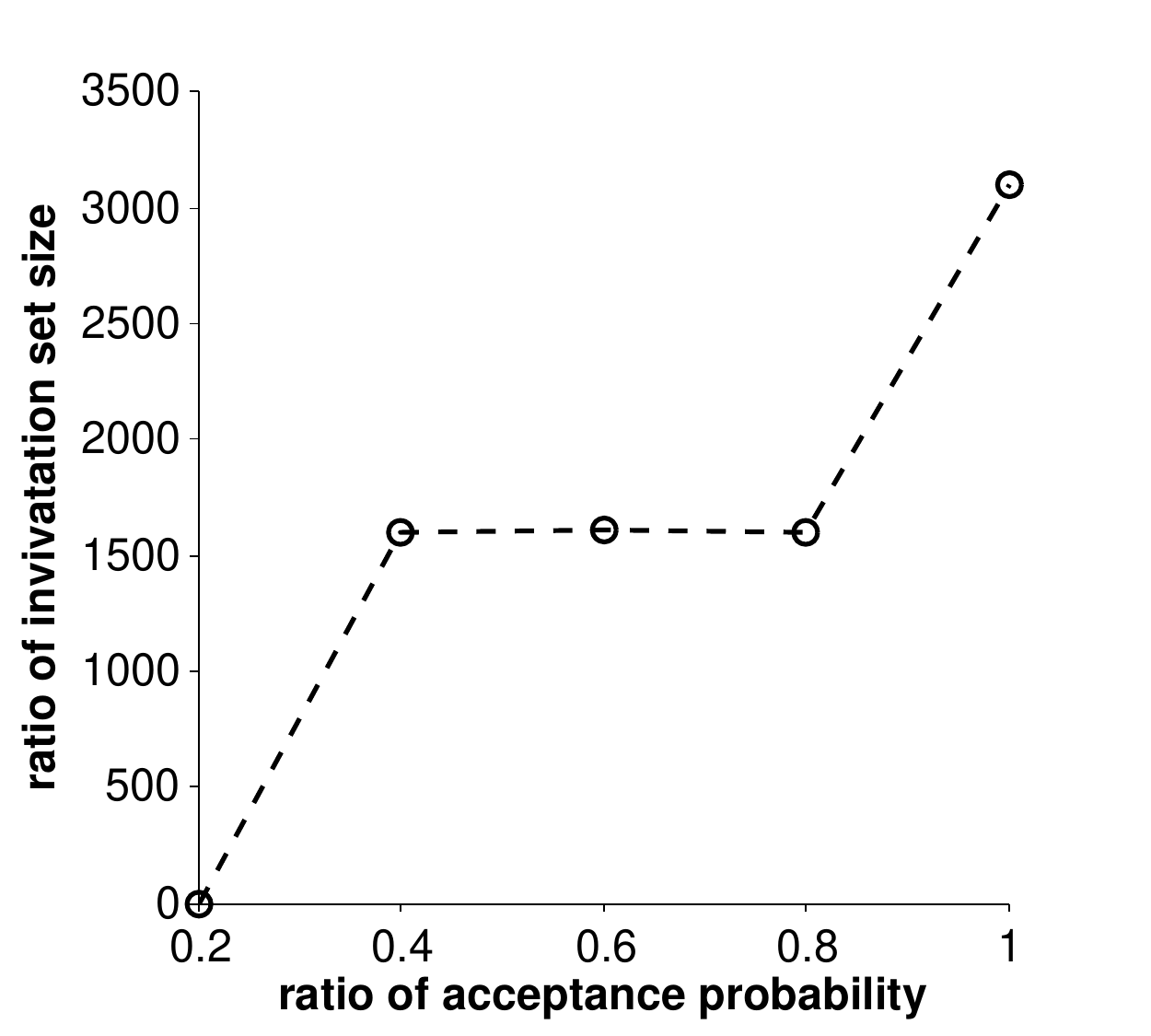}}
\subfloat[Youtube]{\label{fig: youtube_hd}\includegraphics[width=0.24\textwidth]{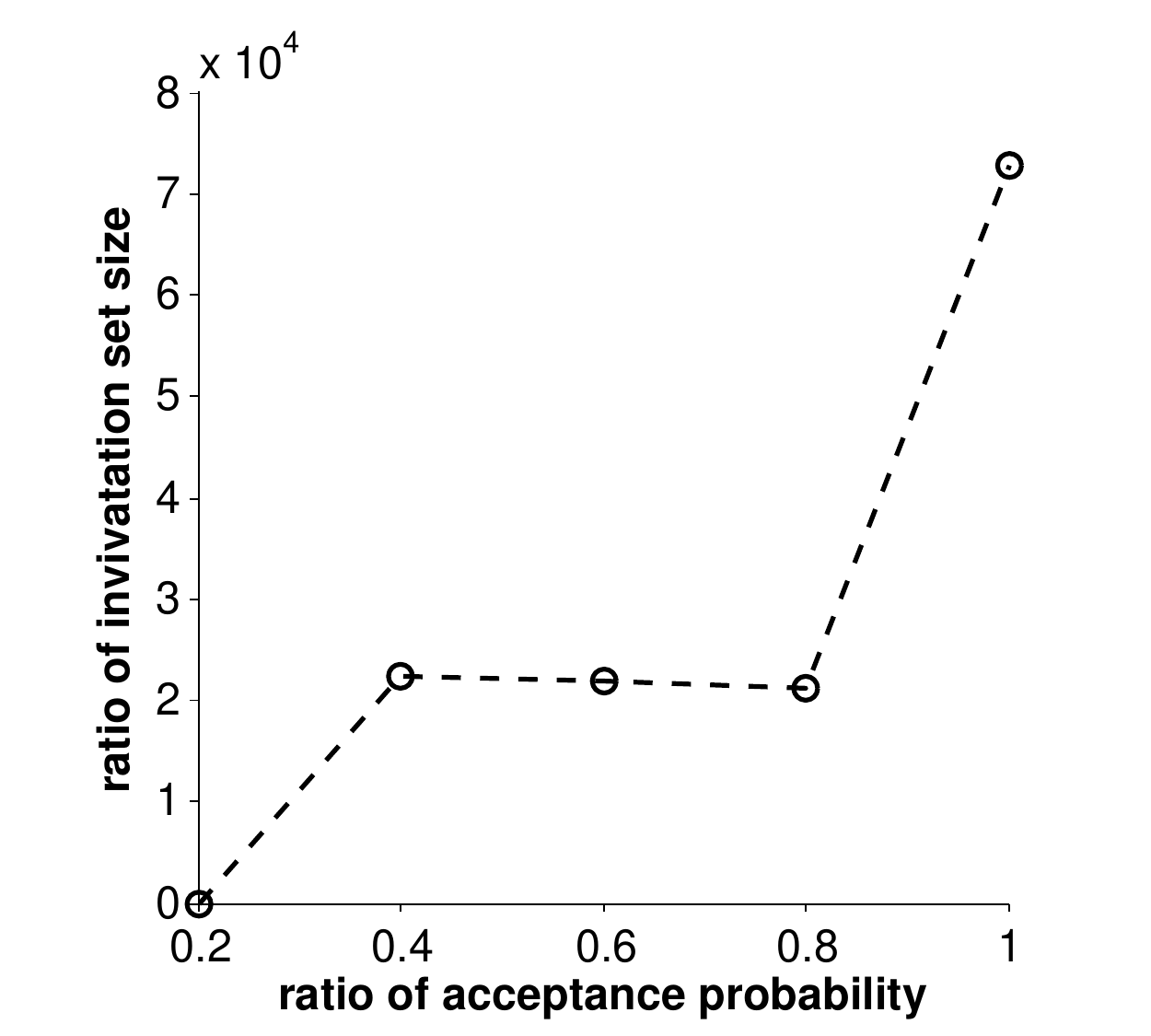}}
\caption{Comparing with HighDegree}
\label{fig: exp_hd}
\end{figure*}

\subsection{Basic Experiment}
\label{subsec: one}
\textbf{Setting.}  We set that $\epsilon=0.01$ and set $N$ as $100,000$ to make the success probability of the algorithm larger than $99.9\%$. For each pair of $s$ and $t$, we first run the RAF algorithm and obtain an invitation set $I_{RAF}$, and then run HD and SP to obtain two invitations with the same size as $I_{{RAF}}$. In particular, until the invitation set reaches the given size, SP fills the invitation set by adding the nodes on the shortest paths from $s$ to $t$, and HD selects the nodes with the highest degree. If more invited nodes are needed, SP will select the next shortest path disjoint from those have been selected.

\textbf{Observations.} The results are shown in Fig. \ref{fig: exp}. The main observation is that even with a very small $\alpha$ the RAF algorithm will produce an invitation set which is able to result in an acceptance probability close to $p_{max}$.\footnote{Note that here $f(I_{RAF})$ is strictly less than $p_{max}$ but the diffidence between them is extremely small.} In addition, when the size of the invitation set is fixed, RAF consistently outperforms HD and SP. On Wiki, as shown in Fig. \ref{fig: wiki_exp_1}, the average acceptance probability resulted by RAF is 0.047, whereas this number is 0.031 under SP or HD. On the other three datasets, SP performs slightly worse than RAF does, while HD cannot produce an effective invitation set. Recall that the friending process can succeed only if the invitation set can connect $s$ and $t$. Thus, SP can at least maintain the connectivity between $s$ and $t$, while HD can hardly do the same on large datasets.

\subsection{Comparing with HD}
\label{subsec: hd}

\textbf{Setting.} 
Following the setting in Sec. \ref{subsec: one}, for each pair of $s$ and $t$, we again first run RAF to obtain $I_{RAF}$. And then we run $HD$ and continuously increase the size of invitation set until the resulted acceptance probability is equal to $f(I_{{RAF}})$. We aim at comparing the size of the invitation sets required by different algorithms to reach the same friending probability. 

\textbf{Observations.} The results are shown in Fig. \ref{fig: exp_hd}, where the $x$-axis denotes the ratio $f(I_{{HD}})/f(I_{{RAF}})$ and the $y$-axis denotes the ratio $|I_{{HD}}|/|I_{{RAF}}|$. We divide the ratio  $f(I_{{HD}})/f(I_{{RAF}})$ into five intervals and compute the average $|I_{{HD}}|/|I_{\text{RAF}}|$ among all the results falling in the same interval. For example, the point $(0.4,5)$ in Fig. \ref{fig: wiki_hd} shows that when the ratio $f(I_{{HD}})/f(I_{{RAF}})$ is around $0.4$, the average of $|I_{{HD}}|/|I_{{RAF}}|$ is close to $4$. 

According to Fig. \ref{fig: exp_hd}, on Wiki, HD requires five times more invited nodes in order to achieve the same acceptance probability resulted by RAF. On HepPh and HepTh,  $|I_{{HD}}|/|I_{{RAF}}|$ is around 3,000 when $f(I_{{HD}})/f(I_{{RAF}})$ is closed to 1. On Youtube, the superiority of RAF becomes more significant under this measure.

The results of different datasets also exhibit different patterns. On Wiki and HepPh, a breakpoint occurs at $x=0.4$ showing that not many new nodes are needed to make the ratio $f(I_{{HD}})/f(I_{{RAF}})$ increases from $0.4$ to $1$. The patterns resulted by HepTh and Youtube are very similar to each other, except that Youtube has a larger scale of the $y$-axis. 

Let us briefly discuss that when a breakpoint may occur. For a particular pair of $s$ and $t$, the breakpoints may occur when there are few paths from $s$ and $t$ and the paths are almost disjoint. Suppose there are only two disjoint paths from $s$ to $t$. After the first path is included in the invitation set, the friending probability cannot increase when more nodes are invited unless the whole second path is included, which results in a sudden increase of the curve.

\begin{figure*}[t]
\centering
\subfloat[Wiki]{\label{fig: wiki_sp}\includegraphics[width=0.23\textwidth]{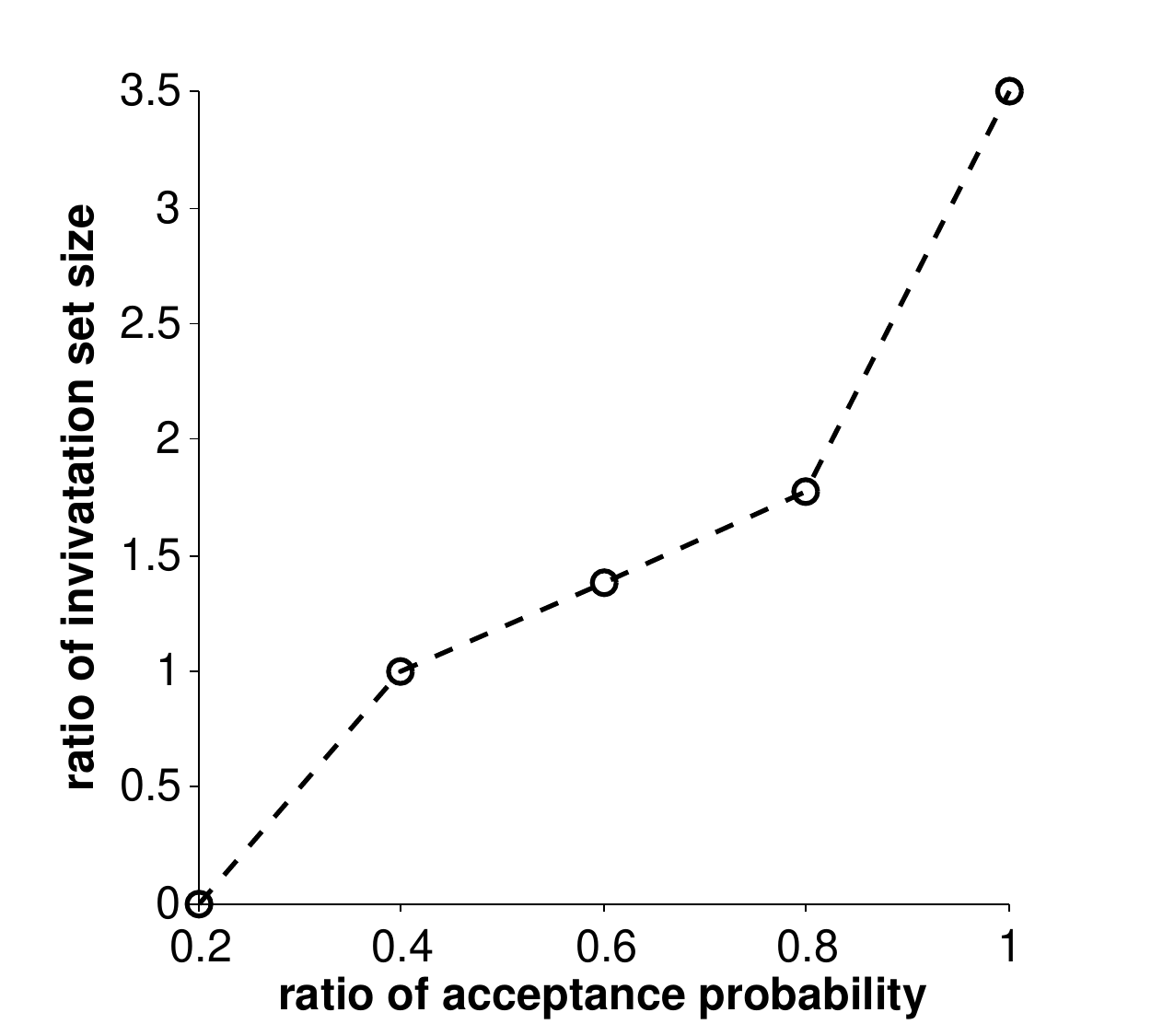}} 
\subfloat[HepPh]{\label{fig: hepph_sp}\includegraphics[width=0.23\textwidth]{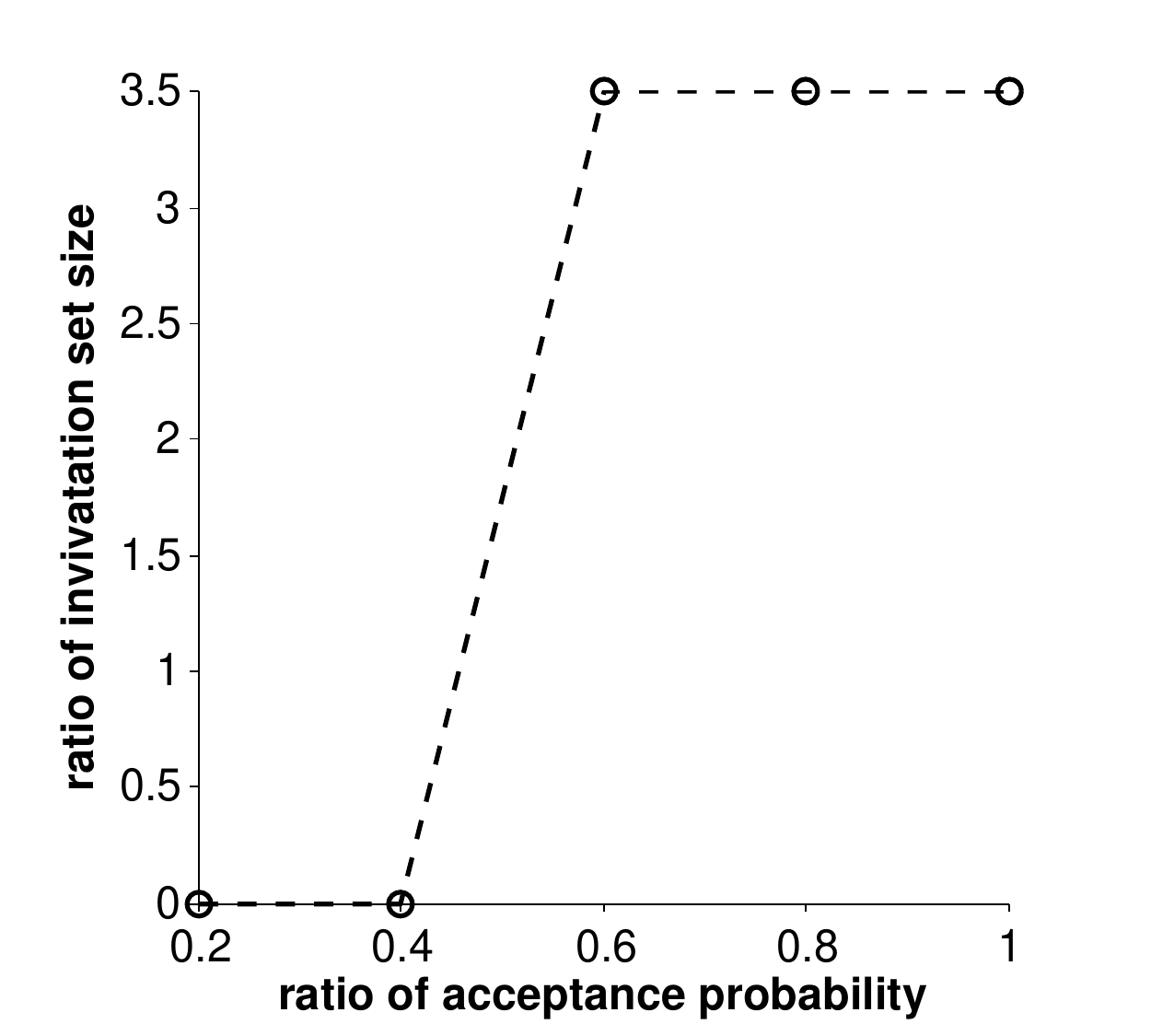}} 
\subfloat[HepTh]{\label{fig: hepth_sp}\includegraphics[width=0.23\textwidth]{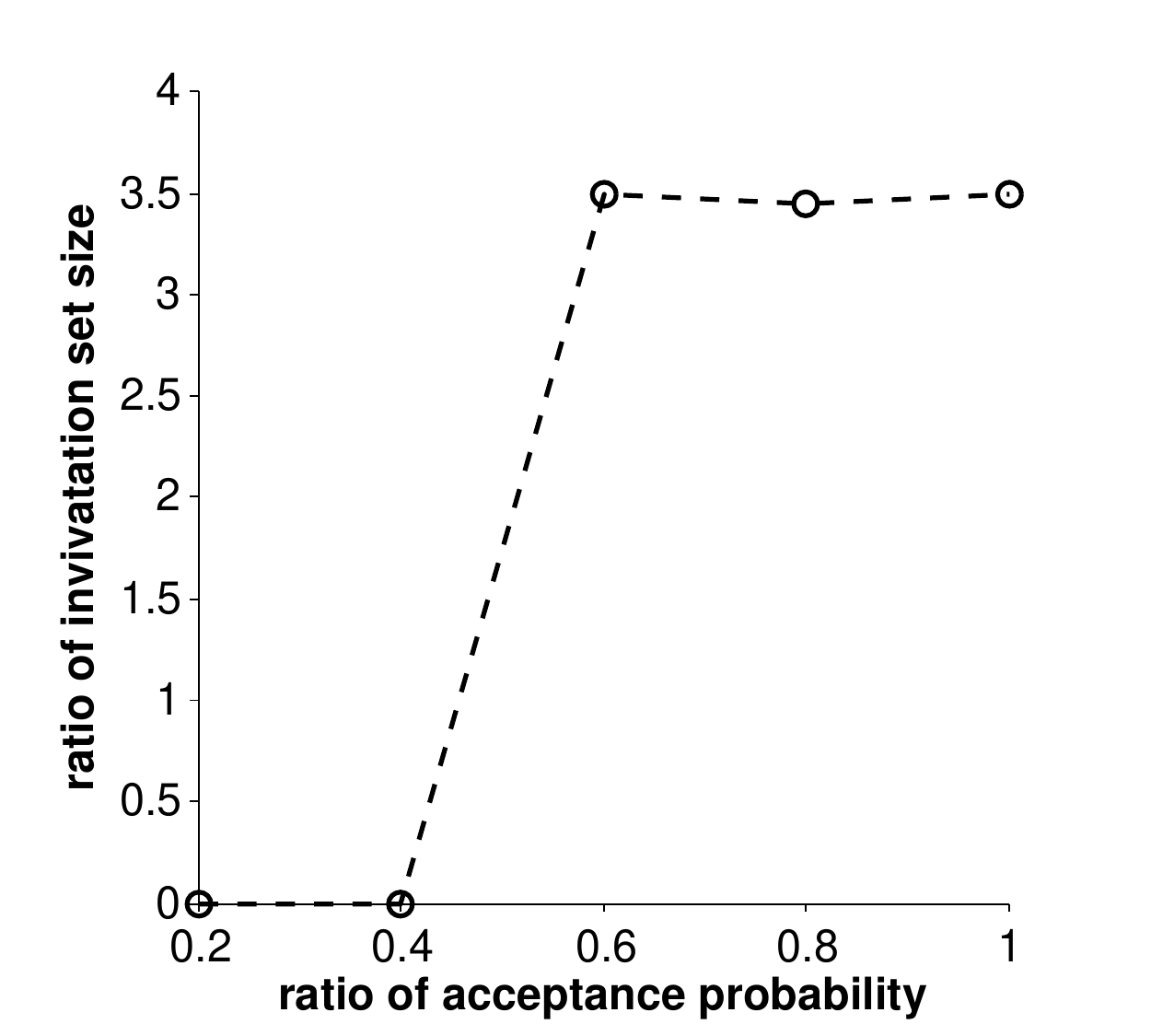}}
\subfloat[Youtube]{\label{fig: youtube_sp}\includegraphics[width=0.23\textwidth]{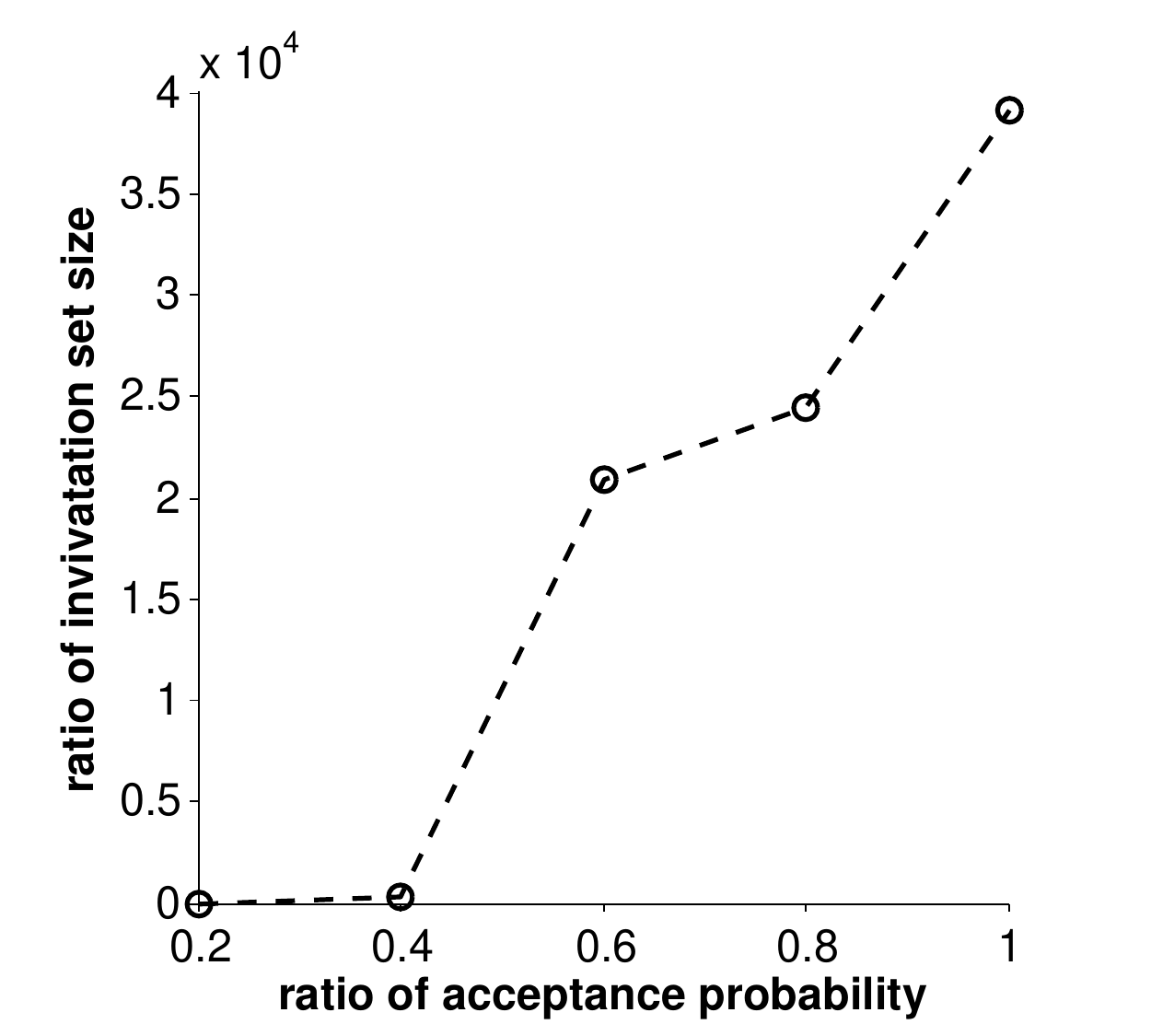}}
\caption{Comparing with ShortestPath}
\label{fig: exp_sp}
\end{figure*}

\begin{table*}[t]
\renewcommand{\arraystretch}{1.5}
\centering
{\begin{tabular}{| C{2.7cm} || C{2.7cm}| C{2.7cm}| C{2.7cm}| C{2.7cm} |}
\hline
\textbf{} 	& \textbf{Wiki} 		& \textbf{HepTh }& \textbf{HepPh} &  \textbf{Youtube} \\   
\hline
\hline
\textbf{Avg. $|V_{max}|$} 	& 130.80		& 165.61		& 915.17 &  6472.21\\
\hline
\textbf{Avg. $|I_{RAF}|$}	& 37.06	& 52.78	& 513.38 &  2126.56 \\
\hline
\textbf{Avg. $|V_{max}|/|I_{RAF}|$}	& 3.45		& 3.89	& 2.63  &  32.77 \\
\hline
\end{tabular}}
\caption{Comparing with $V_{max}$}
\label{table: exp_3}
\vspace{-3mm}
\end{table*}

\subsection{Comparing with SP} 
\label{subsec: sp}

\textbf{Setting.} 
The setting here is similar to that in Sec. \ref{subsec: hd}, except that now we compare RAF with SP.

\textbf{Observations.} The results of this part are shown in Fig. \ref{fig: exp_sp}. On all the three datasets that are relatively small, the number of invited nodes required by SP is less than four times more than $|I_{RAF}|$ in order to achieve $f(I_{RAF})$, indicating that SP is not as good as RAF but still not a very poor heuristic method. However, on Youtube, it requires up to 8,000 times more invited nodes than RAF does to achieve the acceptance probability of $f(I_{RAF})$. Such an observation may suggest that on large graphs a single path is not that relevant for achieving a high acceptance probability due to the fact that a single path can be very long on a large graph and thus the acceptance probability along any single path is not high. Therefore, the overlap between these paths become essential, but SP cannot take account of the dependence between paths. From this perspective, the results herein demonstrate that RAF can better handle large graphs with complex structures of the paths.

\subsection{Comparing with $V_{max}$} 
\label{subsec: exp_3}
\textbf{Setting.} As noted in Lemma \ref{lemma: mini_p_max}, $V_{max}$ is the minimum invitation set that gives $p_{max}$. On the other hand, according to Sec. \ref{subsec: one}, RAF can produce a solution $I_{RAF}$ resulting in an acceptance probability close to $p_{max}$. Therefore, if $|V_{max}|$ is close to $|I_{RAF}|$, the algorithm is not interesting as we can simply take $V_{max}$ as the solution which can be easily computed. In order to figure out this issue, we recorded $V_{max}$ and compared it with the solution of RAF when $\alpha=0.1$. Note that RAF has already been able to produce a good solution when $\alpha$ is equal to $0.1$. For each dataset, we report the average among all tested pairs. 

\textbf{Observations.} The results are listed in Table \ref{table: exp_3}. On Wiki and HepTh, $|V_{max}|$ is at least three times as $|I_{RAF}|$. On Youtube, the average of $|V_{max}|/|I_{RAF}|$ is more than 30. In particular, $V_{max}$ in average requires 6472 invited nodes to achieve $p_{max}$ while $I_{RAF}$ in average needs 2126 invited nodes to produce an acceptance probability close to $p_{max}$. In short, RAF is indeed an effective algorithm in terms of acceptance probability, and it is also efficient concerning the input-output ratio $f(I)/|I|$.

\begin{figure}[!t]
\begin{center}
\includegraphics[width=0.35\textwidth]{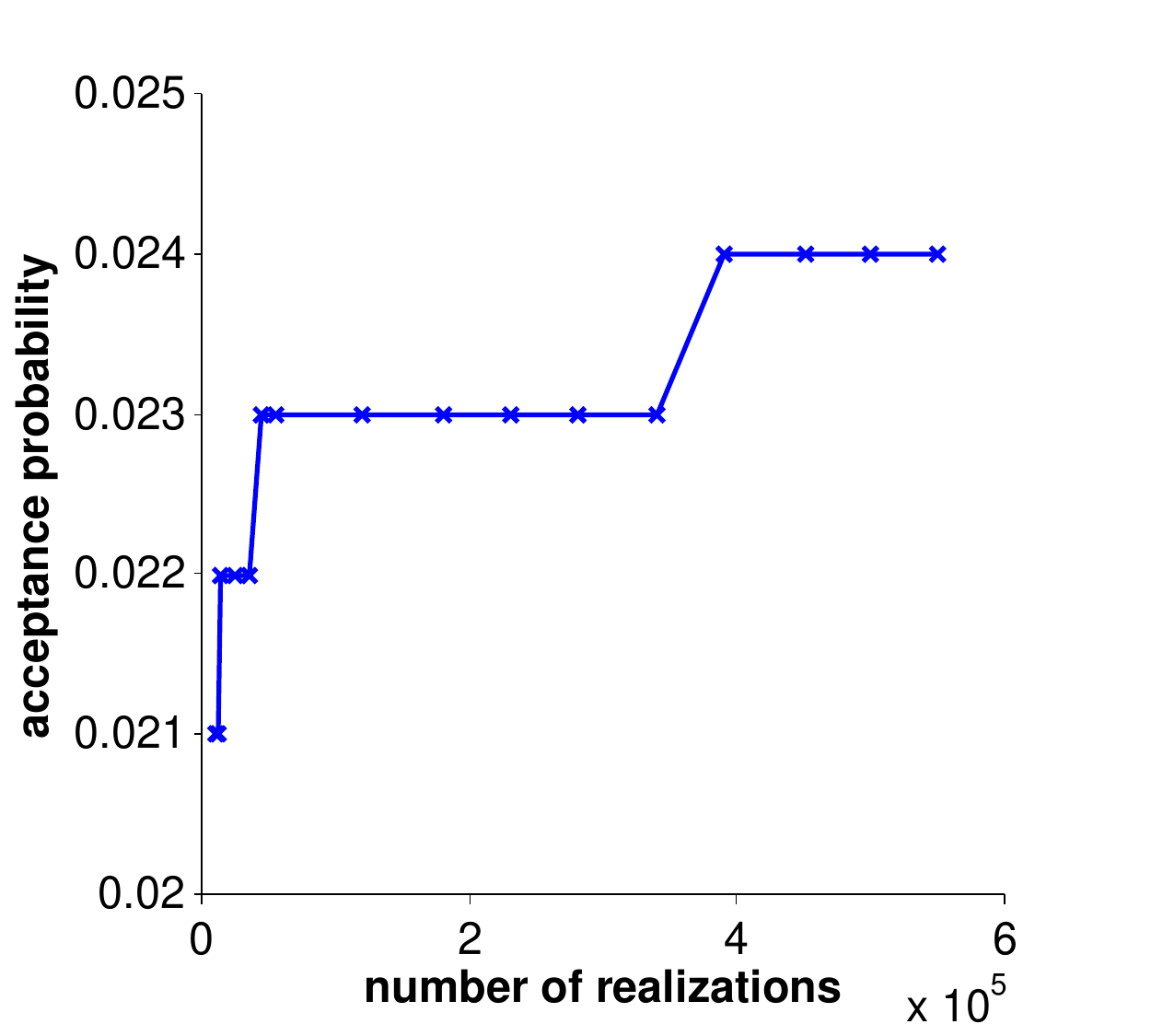} 
\end{center} 
\caption{An illustrative example.}
\label{fig: exp4}
\vspace{-6mm}
\end{figure}

\subsection{Further Discussion}
\label{subsec: further}
According to the analysis in Sec. \ref{sec: algorithm}, the performance of RAF depends on the number of random realizations generated in line 2 in Alg. \ref{alg: algorithm}. In this paper, we give a lower bound which guarantees that the performance can be bounded. However, is this lower bound tight in practice or we have overly generated more realizations than it is needed to reach the maximal performance? We briefly examine this issue by testing different $l$ used in Alg. \ref{alg: algorithm}. One illustrative example collected from Wiki is given in Fig. \ref{fig: exp4}. For this pair of $s$ and $t$, the solution produced by RAF utilized 550,567 invited nodes and the resulted acceptance probability is 0.024. Now we fix $\beta$ and reduce the number the used realizations, and test the acceptance probability resulted by the solution given by Alg. \ref{alg: algorithm}. As shown in the graph, 400,000 realizations would be sufficient to achieve $0.024$. Thus, the running time of RAF can be improved by reducing the number of realizations in practice without hurting the performance. In addition, when only 10,000 realizations are used, the resulted acceptance probability is 0.021 which is not very far from 0.024, which suggests that, in terms of the input-output ratio, a more efficient solution can be found by further reducing the number of realizations used in RAF. We note that these observations apply to many tested pairs and the illustrative example in Fig. \ref{fig: exp4} is not an outlier.

\section{Related Work} 
\label{sec: related}
The existing works primarily focus on the friend recommendation problem. In \cite{chen2009make}, the authors designed several people recommendation algorithms to help users find known offline contacts and discover new friends on Beehive\footnote{Beehive is an enterprise social networking site within IBM.}. A friend recommendation framework to improve recommending quality by characterizing user interest in several dimensions was later studied in \cite{xie2010potential}. The work \cite{agarwal2013collaborative} studied the friend recommendation problem from the view of interaction intensity by using the technique of collaborative filtering. The authors of \cite{hannon2010recommending} also utilized collaborative filtering and considered the problem of recommending twitter users to follow. In \cite{chu2013friend}, the authors proposed another friend recommendation approach with the consideration of real-life location and dwell time. Different from the above works, our paper considers the active friending problem where we aim at building a friendship between an initiator and a specified target user. 

The active friending problem was proposed in \cite{yang2013maximizing} where the friending process was modeled based on the cascade model. Based on an approximate IC model, called MIA \cite{chen2010scalable}, the authors in \cite{yang2013maximizing} studied a simplified problem and designed three algorithms: Range-based Greedy (RG) algorithm, Selective Invitation with Tree Aggregation (SITA) algorithm, and Selective Invitation with Tree and In-Node Aggregation (SITINA) algorithm. Following this line, the authors in \cite{chen2014could} studied the same problem but considered the case when the network forms a DAG. Recently, the authors in \cite{yuan2017active} considered the maximum active friending problem under the linear threshold model and provided an algorithm with a data-dependent approximation ratio by using the super-differential. The linear threshold model has not been widely considered for the active friending problem, though this model has drawn much attention in social network analysis (\cite{chen2010scalable}, \cite{goyal2011simpath}, \cite{he2012influence}, \cite{pathak2010generalized}, \cite{lu2012complexity}). The threshold model has the advantage in modeling the influence of mutual friends on the friending process, which is the main reason that we adopt this model. In addition, it is worthy to note that the active friending problem under the linear threshold model is markedly different from that under the independent cascade model. This problem is neither submodular nor supermodular under the independent cascade model \cite{yang2013maximizing}, while it becomes supermodular under the linear threshold model as shown in \cite{yuan2017active}.

\section{Conclusion}
\label{sec: conc}
In this paper, we study the active friending problem in online social networks. We consider the linear threshold model and design the RAF algorithm with provable performance guarantees. The performance of the proposed algorithm is supported by encouraging experimental. 

One promising future work is to customize the active friending problem for specific social networks, e.g., Facebook, Twitter and LinkedIn. Based on the friending model tailored to different social networks, solutions to active friending are expected to have higher practicability and effectiveness. Second, the approximation hardness of the active friending problem under the linear threshold model is still open. Finally, as noted in Sec. \ref{subsec: further}, it is interesting to further investigate how to reduce the running time of RAF without sacrificing the performance bound. 

\section*{Acknowledgment}
This work is supported in part by the start-up grant from the University of Delaware and the US National Science Foundation under Award \#1747818.

\appendix
\section{Proofs}
\label{appendix: a}
\subsection{Proof of Lemma \ref{lemma: process}}
Note that $f(I)$ is the probability that $t \in C_{\infty}(I)$ and $\E[f(\g,I)]$ is the probability that $t \in H_{\infty}(\g,I)$. It suffices to show that $H_{\infty}(\g,I)$ and $C_{\infty}(I)$ have the same distribution. Because the thresholds are independent from Process 1, the threshold can be generated during the process of generating $C_{\infty}(I)$. Similarly, we can generate the realization along with Process 2. Due to the update rules Eqs. (2) and (3) and the fact that $C_0=H_{0}$, it further suffices to prove that $\Phi(C_i(I))$ and $\Psi(H_{i}(\g,I))$ have the same distribution under the condition that $C_j(I)=H_{j}(g,I)$ for $j<i$. Let us first consider $\Phi(C_i(I))$. For each $u \notin C_i(I)$ and $u \in I$, according to the distribution of $\theta_u$ and Eq. (1), the probability that $u \in \Phi(C_i(I))$ is 
\begin{eqnarray*}
&&\Pr[\theta_u<\sum_{v \in C_i(I)} w_{(v,u)} | \theta_u>\sum_{v \in C_{i-1}(I)} w_{(v,u)}]\\ 
&=&\frac{\sum_{v \in C_i(I)\setminus C_{i-1}(I)} w_{(v,u)}}{1-\sum_{v \in C_{i-1}(I)} w_{(v,u)}}
\end{eqnarray*}
Second, let us consider $\Psi(H_{i}(g,I))$. For each $u \notin H_i(g, I)$ and $u \in I$, according to Def. 1 and Eq. (4), the probability that $u \in \Psi(H_{i}(g,I))$ is $\Pr[g(u) \in H_{i}(\g,I) |g(u) \notin H_{i-1}(g,I)]$ which is
\begin{equation*}
\frac{\sum_{v \in H_{i}(g,I)\setminus H_{i-1}(g,I)}w(v,u)}{1-\sum_{v \in H_{i-1}(g,I)}w(v,u)}.
\end{equation*} 
By the inductive hypothesis, we have 
\begin{equation*}
\frac{\sum_{v \in C_i(I)\setminus C_{i-1}(I)} w_{(v,u)}}{1-\sum_{v \in C_{i-1}(I)} w_{(v,u)}}=\frac{\sum_{v \in H_{i}(g,I)\setminus H_{i-1}(g,I)}w(v,u)}{1-\sum_{v \in H_{i-1}(g,I)}w(v,u)},
\end{equation*}
which completes the proof.

\pagebreak
\bibliography{amobib}
\bibliographystyle{IEEEtran}

\end{document}